\documentclass[reqno,12pt,letterpaper]{amsart}

\usepackage{amsfonts}
\usepackage{dsfont}
\usepackage{mathabx}
\usepackage{mathtools}
\usepackage{hyperref}
\usepackage[alphabetic,nobysame]{amsrefs} 

\usepackage{color}
\newtheorem{thm}{Theorem}
\newtheorem{lemma}{Lemma}[section]
\newtheorem{prop}[lemma]{Proposition}
\newtheorem{cor}[lemma]{Corollary}
\newtheorem{remark}{Remark}
\theoremstyle{definition}
\newtheorem{definition}{Definition}
\newcommand{\floor}[1]{\lfloor #1 \rfloor}

\numberwithin{equation}{section}

\newcommand{\RR}{{\mathbb R}}
\newcommand{\NN}{{\mathbb N}}

\newcommand{\ZZ}{{\mathbb Z}}
\newcommand{\PP}{{\mathbb P}}

\newcommand{\Wxe}{W_\omega(x;E)}
\newcommand{\Wlxe}{W_{\omega,L}(x;E)}
\newcommand{\Wxoe}{W_\omega(x_0;E)}
\newcommand{\Wlxoe}{W_{\omega,L}(x_0;E)}
\newcommand{\Thwe}{{\Theta_{\omega,E}}}
\newcommand{\1}{\mathds{1}}

\newcommand{\Tr}{{\operatorname{Tr}}} 
\newcommand{\dist}{{\operatorname{dist}}} 
\newcommand{\red}{{\operatorname{red}}} 
\newcommand{\EHO}{{\1_{\{E\}}(H_\omega)}}
\def\smallsection#1{\smallskip\noindent\textbf{#1}.}

\newcommand{\bfek}{{\textbf{E}^{(k)}}}

\usepackage{bm}

\title{Dynamical Localization for the Singular Anderson Model in $\mathbb{Z}^d$}
\author{Nishant Rangamani}
\address[Nishant Rangamani]{University of California, Irvine}
\email{nrangama@uci.edu}

\author{Xiaowen Zhu}
\address[Xiaowen Zhu]{University of Washington} 
\email{xiaowenz@uw.edu}
\dedicatory{Dedicated to Abel Klein on the occasion of his 75th birthday}
\begin{document}
\maketitle

\section*{Abstract}

We prove that once one has the ingredients of a ``single-energy multiscale analysis (MSA) result'' on the $\mathbb{Z}^d$ lattice, several spectral and dynamical localization results can be derived, the most prominent being strong dynamical localization (SDL). In particular, given the recent progress at the bottom of the spectrum for the $\mathbb{Z}^2$ and $\mathbb{Z}^3$ cases with Bernoulli single-site probability distribution, our results imply SDL in these regimes.

\section{Introduction}

We consider the $d$-dimensional Anderson model, a random Schr{\"o}dinger operator on $\ell^2(\ZZ^d)$ given by:
\begin{equation}
\label{eq: model}
(H_{\omega}\phi)(n):=\sum_{|m-n|=1}(\phi(m)-\phi(n))+V_{\omega}(n)\phi(n).
\end{equation}
Here, the $V_{\omega}(n)$ are independent and identically distributed ($i.i.d.$) real-valued random variables with common distribution $\mu$, $\forall n\in\ZZ^d$. We will assume that $S\subset\RR$, the topological support of $\mu,$ is compact and contains at least two points. The underlying probability space is the infinite product space $(\Omega, \mathcal F, \PP) = (S^{\ZZ^d}, \mathds B(\RR^{\ZZ^d}), \mu^{\ZZ^d})$, where $\mathds B(X)$ denote all Borel sets in $X$. We denote $\omega\in\Omega$ by $\{\omega_n\}_{n\in\ZZ^d}$. Given $\Lambda\subset \ZZ^d$, we denote the restriction of the probability space $(\Omega, \mathcal F, \PP)$ to $\Lambda$ by $(\Omega_\Lambda, \mathcal F_\Lambda, \PP_\Lambda)$.

In this paper, we provide a comprehensive, self-contained proof that extracts localization results from the single-energy multi-scale analysis (MSA) result. In order to properly contextualize this paper, it is necessary to briefly describe some chronological background.

When \( d = 1 \) in \eqref{eq: model}, localization has been extensively studied and is well-understood: see \cite{FS84}, \cite{KS87}, \cite{DK89, K07} for the case when the distribution measure \(\mu\) is absolutely continuous, and \cite{CKM87, JZ19} for the case when \(\mu\) is singular. For \( d > 1 \), the number of approaches drops dramatically. Nevertheless, Multi-Scale Analysis (MSA), originally introduced by \cite{FS84} and significantly improved by \cite{DK89}, still plays a crucial role.

Soon after the MSA had taken a firm foothold in the literature and community, Germinet and De Bi\`evre \cite{GD98} provided an axiomatic treatment of extracting dynamical localization results from an energy-interval MSA by checking the so-called SULE condition that was originally proposed in \cite{DJLS95,DJLS96}. Later, \cite{DS01} improved the result by extracting strong dynamical localization up to a certain order from the same energy-interval MSA. On the one hand, such an energy-interval MSA can be established for many random models when the single-site distribution $\mu$ is absolutely continuous, e.g. \cite{DK89,GK01} and the references therein, so that localization results can be extracted using \cite{GD98,DS01}. On the other hand, for the continuous Bernoulli-Anderson model ($\mu$ is Bernoulli) in high dimensions ($d\geq 2$), only a single-energy MSA is available due to a weak probability estimate, as shown in \cite{BK05}. As a result, additional efforts are required to extract localization information from the single-energy MSA. In \cite{GK12}, Germinet and Klein addressed this issue by introducing a new infinite volume localization description. Nonetheless, while the proof in \cite{BK05, GK12} contains the key ideas, it is not directly applicable to the discrete Bernoulli-Anderson model in high dimensions due to the absence of a quantitative unique continuation principle in the discrete regime.

Recently, inspired by a probabilistic unique continuation principle developed for the $\mathbb{Z}^2$ lattice, Ding and Smart \cite{DS18} obtained the single-energy MSA result with weak probability estimates and proved Anderson localization for the $2d$ discrete Bernoulli-Anderson model, i.e. \eqref{eq: model} with $d = 2$ and $\mu$ being Bernoulli. This work was then extended to the $\mathbb{Z}^3$ lattice by Li and Zhang \cite{LZ19} where the authors also . As with the continuous Bernoulli-Anderson model, no energy-interval MSA is available under these regimes due to the weak probability estimate. Thus it is our aim to tackle this problem and extract (strong) dynamical localization results from the single-site MSA result derived in \cite{DS18, LZ19} by following the method developed in \cite{GK12}. It is worth mentioning that our proof works for arbitrary dimension $d$. If a single-energy MSA can be established at the bottom of the spectrum when $d>3$, or for the entire spectrum when $d=2$, as anticipated by physicists (which remains an open question in the field), then our results would indicate strong dynamical localization in those regions.

While the techniques presented in this paper closely follow the work of \cite{GK12}, there are some features worth mentioning. Firstly, the approach of \cite{GK12} is developed in the continuum and applies to more general operators, leading to technical difficulties that can be avoided or simplified in the discrete setting. For instance, the generalized eigenfunction expansion (GEE) can be constructed more directly in the discrete regime without referring to more general GEE theory and other references (but we still need the BGKM-decomposition theorem \cite[Theorem 15.2.1]{B96II}). Secondly, the extraction of localization from MSA had undergone several revisions \cite{G99, GJ01, GK01, GK06} before \cite{GK12} and resulted in a form that is perceived as user-friendly but difficult to comprehend, as demonstrated in Definition \ref{def: Wxe} and Theorem \ref{thm: Key Theorem} below. Hence, we attempt to offer some clarification on the evolution of the definition that may provide insight into why it is defined and stated in such a way. Finally, as mentioned above, we formulate our results in a more axiomatic way that we hope will provide help for potential use in the future.


\section{Preliminaries and Main results}
\label{sec: preliminary}
\subsection{Preliminaries} For $x\in\ZZ^d$, let $|x|_1=\max\limits_{i=1,\dots,d}{|x_i|}$ and $|x| = (\sum |x_i|^2)^{1/2}$. Let $\langle x \rangle = (1+|x|^2)^{1/2}$. Let $\langle X\rangle^\nu$ denote the  multiplication operator $\langle x\rangle^\nu$ on $\ell^2(\ZZ^d)$. Note it is unbounded if $\nu\geq 0$. Fix some $\nu>d/2$ through out the paper, for $x_0\in \ZZ^d$, let $(T_{x_0}\phi)(x) = \langle x - x_0 \rangle^{\nu} \phi(x)$. 

Let $\Lambda_L(x)=\{y\in\ZZ^d:|y-x| < L/2\}$ and $\Lambda_{L_2,L_1}(x) = \Lambda_{L_2}(x)\setminus\Lambda_{L_1}(x)$. We omit $x$ if it is clear in the context. Let $\Vert \cdot \Vert_{\ell^2(\Lambda)}$ denote the $\ell^2$ norm on $\ell^2(\Lambda)$ for any $\Lambda\subset \ZZ^d$. We omit $\Lambda$ if it is clear in the context. 

Let $\chi_\Lambda$ denote the projection from $\ell^2(\ZZ) \to \ell^2(\Lambda)$. Let $H_{\omega,\Lambda}:=\chi_\Lambda^* H_\omega \chi_\Lambda$ denote the restriction of $H$ to $\Lambda$, and $G_{\omega,\Lambda,E}:=(H_{\omega,\Lambda}-E)^{-1}$. Let $\1_B(x)$ denote the characteristic function of a set $B\subset \RR$, $\RR^d$ or $\ZZ^d$. Let $\1_B(H)$ denote the spectral projection of $H$ to $B\subset \RR$. 

For an operator $A : \ell^2(\ZZ^d) \to \ell^2(\ZZ^d)$, let $\Vert A \Vert$ denote the operator norm and let $\Vert A \Vert_p = \Tr(|A|^p)^{1/p}$ denote the Schatten norm. In particular, $\Vert A \Vert_1$ and $\Vert A \Vert_2$ are the trace and Hilbert-Schmidt norm respectively. In general $\mathcal B(X, Y)$ denote bounded operators from $X$ to $Y$.

Throughout the paper, $C_{\alpha,\beta}$ represents constants only depending on parameters $\alpha,\beta$ that may vary from line to line and $\ell^2$ refers to $\ell^2(\ZZ^d)$ without further emphasis.

\subsection{Main results} Our main result is to extract ``localization'' from ``single-energy MSA result''. In order to be more explicit, we need some preparations:
\begin{definition}[Good boxes, scales]
  We say that:
\begin{enumerate}
  \item The box $\Lambda = \Lambda_L(x_0)$ is $(\omega,E,m)$-regular if $\forall x,y\in\Lambda$ with $|y-x|\geq \frac{L}{100}$, we have
  \[
    |G_{\omega,\Lambda,E}(x,y)| \leq e^{-m|y-x|}.
  \]
  \item The box $\Lambda = \Lambda_L(x_0)$ is $(\omega,E,m,\eta)$-good if $\Lambda$ is $(\omega,m,E)$-regular and
  \[
   \Vert G_{\omega,\Lambda,E}\Vert \leq e^{L^{1-\eta}}.
  \]
  \item The box $\Lambda = \Lambda_L(x_0)$ is $(\omega,E,m,\eta)$-jgood (just as good) if $\Lambda$ is $(\omega,m,E)$-regular and
  \[
   \Vert G_{\omega,\Lambda,E}\Vert \leq 2e^{L^{1-\eta}}.
  \]
  \item The scale $L\in\ZZ$ is $(E,m,\eta,p)$-good if for any $x\in\ZZ^d$, we have
  \[\PP\{\omega:\Lambda_L(x) \text{~is~} (\omega,E,m,\eta)\text{-good}\}\geq 1-L^{-pd}.\]
\end{enumerate}
\end{definition}
\begin{definition}[Single-energy MSA result]
    We say $H_\omega$ has the \textbf{``single-energy MSA result''} on an interval $\mathcal I\subset \RR$ if there are $m_0>0, 0<\eta_0<1, p_0>0$, and some $L_0$ s.t. any scale $L\geq L_0$ is $(E,m_0,\eta_0,p_0)$-good for any $E\in \mathcal I$.
\end{definition}

We are interested in the following types of localization:
\begin{definition}(localization)
  We say $H_\omega$ exhibits
  \begin{enumerate}
    \item Anderson localization (AL) in an interval $I\subset\RR$ if for a.e. $\omega$, $H_\omega$ has pure point spectrum and its eigenfunctions decay exponentially.
    \item Dynamical localization (DL) of order $p$ in $I$ if for a.e. $\omega$,
    \[
      \sup_{t\in \RR} \Vert\langle X\rangle^p e^{-itH_\omega}\1_I(H_\omega)\delta_0\Vert_{\ell^2}<\infty
    \]
    \item Strong dynamical localization (SDL) in expectation of order $(p,s)$ in $I$ if 
    \[
      \mathbb{E} \left\{ \sup_{t\in \RR} \Vert\langle X\rangle^p e^{-itH_\omega}\1_I(H_\omega)\delta_0\Vert_{\ell^2}^s \right\} <\infty
    \]
  \end{enumerate}
\end{definition}
\begin{remark}\label{rmk: SDLtoDLtoAL}
    It is well-known that SDL implies DL by definition and DL implies AL by the RAGE theorem \cite[\S 5.4]{CS87}; but AL does not imply DL, see \cite{DJLS95, DJLS96}.
\end{remark}

Once the ``single-energy MSA result'' is built on some interval $\mathcal I$, our main result below provides a blackbox for people to use to extract SDL, thus DL and AL, on $\mathcal I$. Recall that $\nu>d/2$.

\begin{thm}[SDL]\label{thm: SDL}
  Let $\mathcal{I}\subset \RR$ be a bounded open interval. Assume there is $m_0>0, 0<\eta_0<1, p_0>0$, and some $L_0 = L_0(m_0,\eta_0,p_0,\mathcal I)>0$, s.t. any $L\geq L_0$ is $(E,m_0,\eta_0,p_0)$-good for all $E\in\mathcal{I}$. Then for any $x\in \ZZ^d$, for any $b>0$, for all $s\in \left(0,\frac{p_0 d}{bd + \nu}\right)$, $H_\omega$ exhibits SDL of order $(bd,s)$ on $\mathcal I$, i.e.
    \begin{equation}
        \label{eq-1a}
        \mathbb{E}\left\{\sup\limits_{t\geq 0}  \left\Vert \langle X\rangle^{bd}e^{-itH_\omega} \1_I(H_\omega) \delta_{0}\right\Vert_{\ell^2}^s \right\} \leq C <\infty.
    \end{equation}
As a result, $H_\omega$ also exhibits DL of any order $p>0$ and AL on $\mathcal I$. 
\end{thm}

In particular, since \cite{DS18, LZ19} has derived the ``single-energy MSA result'' when $d = 2,3$ near the bottom of the spectrum with $\frac{1}{2}$-Bernoulli distribution $\mu$
, our result implies SDL in their settings. 

\begin{cor}\label{cor: DL}
Let $d=2,3$. For any $0<p_0<1/2$, there is $E_0>0$, s.t. for any $b>0$, for any $s\in (0,\frac{p_0 d}{bd + \nu})$, $H_\omega$ exhibits SDL of order $(bd,s)$ on $[0,E_0]$.
\end{cor}

\subsection{Key concept and key theorem}\label{subsec: key_concept_key_theorem}
Here we also want to briefly summarize the main idea of the key concept (Definition \ref{def: Wxe}) and theorem (Theorem \ref{thm: Key Theorem}) in the proof of Theorem \ref{thm: SDL} since they may seem unintuitive at first sight. Recall that $\nu>d/2$ is fixed throughout the paper and $T_{a}\phi(x) = \langle x-a\rangle^\nu \phi(x)$. We omit $a$ if $a = 0$. 
\begin{definition}[Generalized eigenvalue/eigenvector]\label{def-4}
  If $H \psi_E = E \psi_E$, $\psi_E \neq 0$, and $\Vert T^{-1}\psi_E\Vert_{\ell^2} <+\infty$, then we say $\psi_E$ is a \textit{generalized eigenfunction (g.e.f.)} of $H$ with respect to  the \textit{generalized eigenvalue (g.e.v.)} $E$. Let $\Theta_{\omega,E}$ denote the set of all g.e.f.'s of $H_\omega$ with respect to  $E$ and set $\tilde \Theta_{\omega,E} = \Theta_{\omega,E}\cup \{0\}$.
\end{definition}

Note that $\langle b \rangle \leq \sqrt{2} \langle a \rangle \langle a - b \rangle$, thus 
\begin{equation}
    \label{eq-2d}
    \Vert T_a^{-1} \Vert \leq 2^{\frac{\nu}{2}} \langle a - b \rangle^{\nu}\Vert T_b^{-1} \Vert. 
\end{equation} As a result,
\begin{equation}
    \label{eq: equiv}
\Vert T^{-1}\psi_E\Vert_{\ell^2}<\infty \quad \Leftrightarrow \quad \Vert T^{-1}_a\psi_E\Vert_{\ell^2}<\infty, \quad \text{~for~any~} a\in \ZZ^d.
\end{equation}

\smallsection{Key concept} Now we can introduce the key concept, originally introduced in \cite{GK06} and further developed in \cite{GK12}, that plays an important role in the proof of localization results. 
\begin{definition}\label{def: Wxe}
  Given $\omega\in\Omega$, $E\in\RR$ and $x\in\ZZ^d$, we define two quantities
  \[
  \Wxe:=\begin{cases}
        \sup\limits_{\psi_E\in\Thwe }\frac{|\psi_E(x)|}{\Vert T_x^{-1}\psi_E\Vert _{\ell^2}},  & \text{if~} \Thwe \neq\emptyset,\\
              0, & \text{otherwise}.
    \end{cases}
    \]
\[ \Wlxe :=\begin{cases}
    \sup\limits_{\psi_E\in\Thwe }\frac{\Vert \psi_E\Vert_{\ell^2(\Lambda_{2L,L}(x))}}{\Vert T_x^{-1}\psi_E\Vert_{\ell^2} },  & \text{if~} \Thwe \neq\emptyset,\\
    0, & \text{otherwise}.
\end{cases}      
\]
\end{definition}
It is well-defined by the argument above. And by the definition of $T_x$, we see 
\begin{equation}
    \label{eq-2c}
\Wxe \leq 1, \qquad \Wlxe \leq \langle L \rangle^\nu.
\end{equation}

 Notice that $\Wxe$ is not a normalized g.e.f. with respect to  $E$ since the denominator changes with $x$ but it plays a similar role. Such normalization techniques are commonly used to prove dynamical localization because one needs certain kinds of uniform control over all generalized eigenfunctions. See for example \cite{G99, GJ01} where normalized generalized eigenfunctions $\widetilde \psi_{E, \omega}(x) = \frac{\psi_{E, \omega}(x)}{\Vert T^{-1}\psi_{E, \omega}\Vert_{\ell^2}}$ are first introduced and used in the proof of dynamical localization. $\Wxe$, $\Wlxe$ are introduced by Germinet and Klein \cite{GK06, GK12} in order to be able to extract all (strong) dynamical localization results from Multi-scale analysis all at once.

\smallsection{Key Theorem} The following theorem is the key to extracting localization from ``single-energy MSA result'', as introduced in \cite{GK12}. Once we have ``single-energy MSA result'', the theorem states that with high probability, if some g.e.f. is subexponentially localized near $x_0$, then all g.e.f. will decay exponentially away from $x_0$ for all $E$. 
\begin{thm}[Key Theorem]\label{thm: Key Theorem}
   Assume there is $m_0$, $p_0>0$, $\eta_0\in (0,1)$ and $\mathcal L$ such that any $L\geq \mathcal L$ is $(E, m_0, \eta_0, p_0)$-good scale for all $E\in \mathcal I$. Then for any $0<p<p_0$, there is $c>0$, $\mu\in (0,1)$, such that when $L$ is large enough, for any $x_0 \in \ZZ^d$, there is an event $\mathcal U_{L,x_0}\in \mathcal F_{\Lambda_L(x_0)}$ such that 
          \begin{equation}\label{eq-2b}
      \PP\{\mathcal{U}_{L,x_0}\}\geq 1-L^{-pd}.
  \end{equation}
  and for all $\omega\in\mathcal{U}_{L,x_0}$, for any $E\in \mathcal I$ with $\dist(E, \mathcal I^c)\geq e^{-cL^\mu}$, we have 
    \begin{equation}\label{Wmain}
    W_\omega(x_0;E)>e^{-cL^{\mu}}\Rightarrow W_{\omega, L}(x_0;E)\leq e^{-cL};
  \end{equation}
  thus 
    \begin{equation}\label{eq-2a}
    W_\omega(x_0;E)W_{\omega, L}(x_0;E)<e^{-\frac{1}{2} c L^\mu}
  \end{equation}
  when $L$ is large enough. 
\end{thm}


\subsection{Sketch of proof} \label{subsec:sketch}The sketch of proof is as follows: We wish to find a large set of configurations $\omega$ (i.e. with high probability) such that $\Wlxe \leq e^{-cL}$, holds for all $E\in \mathcal I$ - Because such uniform version of decaying of $\Wlxe$ would imply certain uniform decay of generalized eigenfunctions, hence strong dynamical localization. 

This is not hard to achieve for a given $E_0$ and its exponentially small neighborhood $|E - E_0| \leq e^{-mL}$, for any $m<m_0$, c.f. Proposition \ref{lemma: fixed_E0}. However, to cover a fixed interval $\mathcal I\subset \RR$, we need to apply Proposition \ref{lemma: fixed_E0} to at least $e^{mL}$-many evenly distributed $E_0$ over $\mathcal I$. But this would destroy the probability estimate since $e^{mL}L^{-pd}>>1$. 

Therefore, we need to control the number of $E_0$ for which we invoke Proposition \ref{lemma: fixed_E0} in order to control the probability from blowing up. This is done by the so-called ``spectral reduction''. 

Very roughly speaking, we want to pick $E\in \sigma(H_{\omega, \Lambda})$, for different scales of $\Lambda$, and only apply Proposition \ref{lemma: fixed_E0} on them. In particular, if $p_0>1$, we only need the ``first spectral reduction'', i.e. Theorem \ref{side_1}, see Remark \ref{rmk-7b}. If $p_0<1$, we will also need the ``second spectral reduction'', i.e. Theorem \ref{side_2}.

The paper is organized as follows:  
\begin{itemize}
  \item Section \ref{sec: preliminary} includes preliminaries, main results, and a sketch of proof.
  \item Section \ref{sec: GEE} introduced the generalized eigenfunction expansion.
  \item Section \ref{sec: localization} extract localization, i.e. Theorem \ref{thm: SDL} from the key Theorem \ref{thm: Key Theorem}. 
  \item Section \ref{sec: Preliminary_Lemmas} and \ref{sec-percolation} made some preparations for the spectral reduction.
  \item Section \ref{sec: two_spectral_reduction} proves Theorem \ref{thm: Key Theorem} using two spectral reductions.
\end{itemize}

\section{Generalized eigenfunction expansion}
\label{sec: GEE}
We give a short introduction of generalized  eigenfunction expansions (GEE) needed for the proof of Theorem \ref{thm: SDL} in Sec \ref{sec: localization}. 

The idea is as follows: Not every self-adjoint operator has a complete eigenbasis, for example those with continuous spectrum. However, if one could enlarge the domain (using rigged Hilbert space) and allow eigenfunctions to be ``generalized eigenfunctions'' (formal eigenfunctions that belong to this larger domain), then every self-adjoint operator could have a diagonal decomposition with respect to  these ``generalized eigenfunctions''. This procedure is rigorously done for general appropriate operators and rigged spaces in \cite[Sec 10.3, 10.4, 14.1, 15.1, 15.2]{B96II} and eventually leads to the so-called BGKM-decomposition or GEE in \cite[Theorem 15.2.1]{B96II}. In the continuous regime, further verification is needed \cite{GK12} , in the discrete regime, we will do most of the construction more directly. It can be directly verified that the construction here coincides with \cite{B96II} thus they are well-defined. We will borrow the Bochner theorem and BGKM-decomposition from \cite[Theorem 15.1.1, 15.2.1]{B96II} without proof.

\smallsection{Rigged spaces} Let $\mathcal H = \ell^2(\ZZ^d, dx)$ with inner product $\langle u, v\rangle_{\ell^2} = \sum_{x\in \ZZ^d} \overline{u(x)}v(x)$. Let $\mathcal H_+, \mathcal H_-$ be weighted-$\ell^2$ spaces:
\[
  \mathcal H_+ = \ell^2(\ZZ^d, \langle x \rangle^{2\nu}dx), 
  \quad
  \mathcal H_- = \ell^2(\ZZ^d, \langle x \rangle^{-2\nu}dx).
\]
with inner product and norm being
\[
\langle u,v\rangle_{\pm} = \sum\limits_{x\in \ZZ^d} \overline{u(x)}v(x)\langle x \rangle^{\pm 2\nu}, \quad \Vert u\Vert_+ = \Vert \langle x \rangle^{\nu}u\Vert_{\ell^2}, \quad \Vert u\Vert_- = \left\Vert \langle x \rangle^{-\nu}u\right\Vert_{\ell^2}.
\]
Since $\Vert \cdot \Vert_- \leq \Vert\cdot \Vert_{\ell^2}\leq \Vert \cdot \Vert_+$, we have $\mathcal H_+ \subset \mathcal H\subset \mathcal H_-$ in a both continuous and dense sense. This chain is a chain of rigged Hilbert spaces, cf \cite[Sec 14.1]{B96II}. The definition there is more general and involved, but it coincides with the $\mathcal H_\pm$ given above in the discrete setting.

Since the embeddings are dense, the inner product $\langle \cdot, \cdot \rangle_{\ell^2}$ defined on $\mathcal H \times \mathcal H$ extends continuously to $\mathcal H_- \times \mathcal H_+$. More specifically, for $u\in \mathcal H_+ = \langle x \rangle^{-\nu}\mathcal H$, $v\in \mathcal H_- = \langle x\rangle^{\nu}\mathcal H$, the formula for the extended inner product (which we still denote as $\langle \cdot, \cdot\rangle_{\ell^2}$) is $\langle u, v\rangle_{\ell^2} = \sum\limits_x \overline{u(x)}v(x)$. 

\smallsection{Operators in $\mathcal B(\mathcal H_+, \mathcal H_-)$} Given a bounded operator $A$ from $\mathcal H_+$ to $\mathcal H_-$, denoted as $A\in \mathcal B(\mathcal H_+, \mathcal H_-)$, we say $A$ is \textit{positive} if $\langle Au,u\rangle_{\ell^2}\geq 0$ for $u\in \mathcal H_+$.  We define the trace of a positive operator $A\in \mathcal B(\mathcal H_+, \mathcal H_-)$ to be 
\[
\Tr_{\pm}(A) := \sum\limits_{n} \langle Au_n,u_n\rangle_{\ell^2}
\]
when the sum is finite. Here $\{u_n\}_n$ is any orthonormal basis (ONB) of $\mathcal H_+$. In particular, let $p(x) = \langle x \rangle^{\nu}$, then $\left\{p(x)^{-1}\delta_x\right\}_{x\in \ZZ^d}$ forms an ONB of $\mathcal H_+$. Thus we can rewrite 
\[
\begin{split}
\Tr_\pm(A) &= \sum\limits_{x\in \ZZ^d} \langle A p(x)^{-1} \delta_x, p(x)^{-1} \delta_x\rangle_{\ell^2}\\
&= \sum\limits_{x\in \ZZ^d} \langle p(x)^{-1} A p(x)^{-1} \delta_x, \delta_x\rangle\\
&= \Tr(p(\cdot)^{-1} A p(\cdot)^{-1}) = \Tr(T^{-1}AT^{-1})
\end{split}
\]
where $\Tr$ denotes the standard trace of a trace class operator in $\mathcal B(\mathcal H, \mathcal H)$ and $T^{-1}u(x) = p(x)^{-1}u(x)$. In other words, $\Tr_\pm(A)$ is well-defined and equals to $\Tr(T^{-1}AT^{-1})$ when $T^{-1}A T^{-1}$ is a trace class operator from $\ell^2(\ZZ^d)$ to $ \ell^2(\ZZ^d)$. This sheds the light of the consideration of $\Tr(T^{-1}f(H)T^{-1})$ in \cite[Assumption GEE, SGEE]{GK01} and \cite[Sec 5.4]{GK12}. We call trace $\Tr$ in $\mathcal B(\mathcal H, \mathcal H)$ by ``trace'' and call $\Tr_\pm$ in $\mathcal B(\mathcal H, \mathcal H)$ by ``$\pm$-trace'' for clarity.

\smallsection{Embeddings and the Bochner Theorem} Let $i_+:\mathcal H_+\to \mathcal H$, and $i_-:\mathcal H\to \mathcal H_-$ be the embedding maps $i_+u = u$, $i_-v = v$. 

Assume $B:\mathcal H\to \mathcal H$ is a bounded operator. It can be easily checked that $B_\pm := i_-B i_+:\mathcal H_+\to \mathcal H_-$ induced by $B$ is also bounded. In particular, fix $\omega\in \Omega$, given any Borel set $I\in \mathds B(\RR)$, we can define
\[
    \1_{I,\pm}(H_\omega):= i_-\1_I(H_\omega)i_+\in \mathcal B(\mathcal H_+, \mathcal H_-)
\]
and check that
\begin{equation}
    \label{eq-3a}
    \begin{split}
    \mu_\omega(I) :=& \Tr_\pm(i_-\1_{I}(H_\omega)i_+)\\
    =&  \Tr(T^{-1}\1_{I}(H_\omega) T^{-1})=\Vert T^{-1}\1_I(H_\omega)\Vert_2^2\\
    =& \sum\limits_x \Vert \1_I(H_\omega)p(x)^{-1}\delta_x\Vert_2^2\\
    \leq  &\sum_x\left\Vert p(x)^{-1}\delta_x\right\Vert^2\leq \sum_x p(x)^{-2}<+\infty,
\end{split}
\end{equation}
Thus we obtain $\{\1_{I, \pm}(H_\omega)\}_I$, a $\mathcal B(\mathcal H_+, \mathcal H_-)$-operator-valued measure with finite $\pm$-trace (see Theorem \ref{thm: Bochner} below for definition). For such operator-valued measure, we recall the Bochner theorem \cite[Theorem 15.1.1]{B96II}.
\begin{thm}[\text{\cite[Theorem 15.1.1]{B96II}}]\label{thm: Bochner}
  Let $\theta: \mathds B(\RR) \to \mathcal B(\mathcal H_+,\mathcal H_-)$ be an operator-valued measure with finite $\pm$-trace, i.e. 
  \begin{enumerate}
    \item $\theta(I)$ is non-negative for any Borel set $I\subset \RR$,
    \item $\Tr_{\pm}(\theta(\RR))<\infty$,
    \item $\theta(\bigsqcup\limits_j I_j) = \sum\limits_j \theta(I_j)$, with convergence in the weak sense.
  \end{enumerate}
  Then $\theta$ can be differentiated with respect to  the trace measure $\rho(I) := \Tr_{\pm}(\theta(I))$ and there exists $P(E): \mathcal H_+ \to \mathcal H_-$ with
  \[
    \begin{cases}
      0\leq P(E) \leq \Tr_\pm(P(E)) = 1,\quad \rho\text{-a.e.~} E,\\
      P(E) \text{~is~weakly~measurable~w.r.t.~} \mathcal B(\RR),\\
      \text{The integral converges in the Hilbert-Schmidt norm.}
    \end{cases}
  \]
  such that 
  \[
    \theta(I) = \int_I P(E) d\rho(E).
  \]
\end{thm}

\smallsection{Generalized eigenfunction decomposition}
By applying Theorem \ref{thm: Bochner} to $\{\1_{I, \pm}(H_\omega)\}_I$, we obtain (see also \cite[Theorem 15.1.2]{B96II}): There exists weakly measurable operators $P_\omega(E):\mathcal H_+ \to \mathcal H_-$, and trace measure $\mu_\omega(I)$ \eqref{eq-3a}, such that
\begin{equation}\label{eq-3e}
  i_- \1_{I}(H_\omega) i_+ = \int_{I} P_\omega(E) d\mu_\omega(E)
\end{equation}
with $\Tr_{\pm}(P_\omega(E)) = 1,\text{~for~}\mu_\omega\text{-a.e.~} E$. 
Furthermore, \cite[Theorem 15.2.1]{B96II} states that for $u\in \mathcal H_+$, $f\in \mathcal B_{1,b}(\RR)$ (bounded Borel function over $\RR$). We have: 
\begin{equation}
  \label{Bochnerf(H)}
  \left(i_- f(H_\omega)i_+\right)u = \left(\int_{B} f(E)P_\omega(E) d\mu_\omega(E)\right) u
\end{equation}\footnote{Note that the notations in Theorem 15.2.1 are abused where they omitted $i_\pm$ in \eqref{Bochnerf(H)}}
and 
\begin{equation}
    \label{eq-3d}
  \text{Range}(P_\omega(E))= \tilde{\Theta}_\omega(E), \quad \text{~for~} \mu_\omega \text{-a.e.} E.
\end{equation}
This is the so-called \textit{generalized eigenfunction expansion (GEE)}, that is, a decomposition of rigged spectral projection $i_-\1_I(H_\omega)i_+$ with respect to ``generalized eigenspaces'' $\widetilde \Theta_{\omega, E}$.

In particular, if we can show $\text{Range}(P_\omega(E))\subset \mathcal H$, for $\mu_\omega$-a.e. $E$ in an interval $\mathcal I$, then every generalized eigenfunction become an eigenfunction; thus $H_\omega$ has pure point spectrum on $\mathcal I$.

\smallsection{Pure point case}
Given any $E\in \RR$, by \eqref{eq-3e} and \eqref{eq-3a},  
\begin{equation}
    \label{eq-3k}
i_- \1_{\{E\}}(H_\omega)i_+ = P_\omega(E) \mu_\omega(\{E\}) = P_\omega(E)\Vert T^{-1} \EHO\Vert_2^2.
\end{equation}
Assume $H_\omega$ has pure point spectrum in an interval $I$, denoted by $\{E_i\}$. By \eqref{Bochnerf(H)} and \eqref{eq-3k}, 
\[
\begin{split}
  i_-f(H_\omega) \1_I(H_\omega)i_+  &= \int_I f(E) P_\omega(E) d\mu_\omega(E) = \sum_{i} f(E_i) P_\omega(E_i) \mu_\omega(\{E_i\})\\
  &= \sum_i f(E_i) i_- \1_{\{E_i\}}(H_\omega) i_+ = i_- \int_I f(E) \EHO \frac{d\mu_\omega(E)}{\mu_\omega(\{E\})}i_+.
\end{split}
\]
Thus we obtain a decomposition in $\mathcal B(\mathcal H, \mathcal H)$ when $H_\omega$ has pure point spectrum:
\begin{equation}
    \label{eq-3j}
    f(H_\omega) \1_I(H_\omega) = \int_I f(E) \EHO \frac{d\mu_\omega(E)}{\mu_\omega(\{E\})}.
\end{equation}
\begin{remark}
    Notice that \cite[Sec 5.4]{GK12} introduced $\Wxe$ and $\textbf{W}_{\omega}(x;E)$ where $\Wxe$ take the supremum over all generalized eigenfunction $\Theta_\omega(E)$ while $\textbf{W}_\omega(x;E)$ takes the supremum over all functions in $\text{Range}(P_\omega(E))$. But this is not necessary in our setting because we have \eqref{eq-3d} from \cite[Theorem 15.2.1]{B96II}, while \cite{GK12} only used $\Theta_\omega(E) \subset \text{Range}(P_\omega(E))$ from \cite{KlKS02}.
\end{remark}

\section{Proof of Theorem \ref{thm: SDL}}
\label{sec: localization}
In this section, we extract localization results, i.e. Theorem \ref{thm: SDL}, from Theorem \ref{thm: Key Theorem}. 

\begin{lemma}\label{lemma-4a}
    Under the same assumption as Theorem \ref{thm: Key Theorem}, for any $p<p_0$, for any open interval $I\subset \bar{I} \subset \mathcal I$, for any $s<\frac{pd}{\nu}$, for any $r<pd - s\nu$,  there is $C$ such that 
    \begin{equation}
        \label{eq-3g}
        \mathbb{E} \left\{\left\Vert W_\omega(x;E)W_{\omega, L}(x;E)\right\Vert_{L^\infty(I,d\mu_\omega(E))}^s \right\} \leq CL^{-(pd - s\nu)}.
    \end{equation}
    As a consequence, $H_\omega$ has pure point spectrum. 
\end{lemma}

\begin{proof}
    First, take $L$ large enough such that $I\subset \bar{I}\subset \mathcal I$. By \eqref{eq-2b}, \eqref{eq-2a} in Theorem \ref{thm: Key Theorem} and \eqref{eq-2c}, we have
    \[
    \begin{split}
      \mathbb{E} \left\{\left\Vert W_\omega(x;E)W_{\omega, L}(x;E)\right\Vert_{L^\infty(I,d\mu_\omega(E))}^s \right\}&\leq Ce^{-\frac{s}{2}cL^{\mu}}\PP\{\mathcal{U}_{0,L}^c\}+C\langle L\rangle^{s\nu}\PP\{\mathcal{U}_{0,L}\}\\
      &\leq Ce^{-\frac{s}{2}cL^{\mu}}+C(2L)^{s\nu}L^{-pd}\\
      &\leq CL^{-(pd - s\nu)}
    \end{split}
    \]
    for all $L$ when we take $C$ to be large enough constant (recall that constant $C$ may vary from line to line). As a result, when $s<\frac{pd}{\nu}$, $r<pd - s\nu$, we have 
    \[
    \mathbb{E}\left\{ \sum_{k = 0}^\infty 2^{kr}\left\Vert\Wxe W_{\omega, 2^{k}}(x;E) \right \Vert^s_{\ell^\infty(I, d\mu_\omega)} \right \} \leq \sum_{k = 0}^\infty 2^{-k(pd - s\nu)}<+\infty
    \]
    As a result, for $\PP$-a.e. $\omega$, 
    \[\sum_{k = 0}^\infty 2^{kr}\left\Vert \Wxe W_{\omega, 2^{k}}(x;E) \right \Vert^s_{\ell^\infty(I, d\mu_\omega)} <+\infty.
    \]
    Thus for $\PP$-a.e. $\omega$, there is $C = C_\omega$ such that
    \[
    \Vert \Wxe W_{\omega, 2^{k}}(x;E)\Vert_{\ell^\infty(I, d\mu_\omega)} \leq C2^{-kr/s}.
    \]
    As a result, given any generalized eigenfunction $\psi_E \in \Theta_\omega(E) \subset \text{Range}(P_\omega(E))$, by Definition \ref{def: Wxe}, \eqref{eq-2d} and \eqref{eq-3a}, we have
    \begin{equation}
        \label{eq-3f}
    \begin{split}
    |\psi_E(x)|\Vert\psi_E\Vert_{\ell^2(\Lambda_{2^{k +1}, 2^k}(x))}&\leq \Vert T_x^{-1} \psi_E \Vert_{\ell^2} |\Wxe W_{\omega, 2^{k}}(x;E)|\\
    &\leq C 2^{-kr/s} \Vert T_x^{-1} P_\omega(E) \Vert_2^2 \\
    &\leq C2^{-kr/s} \langle x \rangle^\nu \Vert T^{-1} P_\omega(E)\Vert_2^2 \leq C_{x}2^{-kr/s} 
    \end{split}
    \end{equation}
    for $\mu_\omega$-a.e. $E$. Since $\psi_E \neq 0$, there is some $x_0$ such that $\psi_E(x_0)\neq 0$. Apply \eqref{eq-3f} to $x_0$ and sum up over $k$ from $0$ to $\infty$, we see 
    \[
    \Vert \psi_E\Vert_{\ell^2} = |\psi_E(x_0)| + \sum_{k = 0}^\infty \Vert \psi_E\Vert_{\ell^2(\Lambda_{2^{k +1}, 2^k}(x_0))} \leq \frac{C_{x_0}}{|\psi_E(x_0)|} \sum_{k = 0}^\infty e^{-kr/s}<+\infty.
    \]
    Thus each generalized eigenfunction becomes an eigenfunction; hence for $\PP$-a.e. $\omega$, $H_\omega$ has pure point spectrum. 
\end{proof}

\begin{proof}[Proof of Theorem \ref{thm: SDL}]
Since $H_\omega$ has pure point spectrum $\PP$-a.e. $\omega$, for such $\omega$, recall that we have 
\[
\begin{split}
     \Vert \1_{\Lambda_{2L, L}(0)} f(H_\omega) \1_{I}(H_\omega)\delta_0 \Vert_1 &\leq \int_I f(E) \Vert \1_{\Lambda_{2L, L}(0)}\1_{\{E\}}(H_\omega)\delta_0\Vert_1\frac{d\mu_\omega(E)}{\mu_\omega(\{E\})}. 
\end{split}
\]
By Definition \ref{def: Wxe} and \eqref{eq-3a},
\[
\begin{split}
    \Vert \1_{\Lambda_{2L, L}(0)} \1_{\{E\}}(H_\omega)  \delta_0 \Vert_1 &\leq \Vert \1_{\Lambda_{2L, L}(0)} \EHO \Vert_2 \Vert \delta_0 \EHO \Vert_2\\
    &\leq |W_\omega(0;E) W_{\omega,L}(0;E)| \Vert T^{-1} \EHO \Vert^2\\
    &\leq C |W_\omega(0;E) W_{\omega,L}(0;E)|.
\end{split}
\]
Hence
\[
\Vert \1_{\Lambda_{2L, L}(0)} f(H_\omega) \1_I(H_\omega) \delta_0\Vert_1 \leq \Vert f \Vert_{L^\infty(I, d\mu_\omega)} \Vert W_\omega(0;E) W_{\omega, L}(0;E) \Vert_{L^\infty(I,d\mu_\omega)}
\]
Therefore, 
\[
\begin{split}
\Vert \langle X \rangle^{bd}f(H_\omega) \1_{I}(H_\omega)\delta_0 \Vert^s_1 &\leq C_f \sum_{k = 0}^\infty \langle 2^{k+1} \rangle^{sbd}  \Vert \1_{\Lambda_{2^{k + 1}, 2^k}(0)}  f(H_\omega) \1_{I}(H_\omega)\delta_0\Vert^s_1\\
&\leq C_f \sum_{k = 0}^\infty 2^{ksbd} \Vert W_\omega(0;E) W_{\omega, 2^k}(0;E)\Vert_{L^\infty(I, d\mu_\omega)}^s
\end{split}
\]
Now given any $b>0$, $0<s< \frac{p_0d}{bd + \nu}$, we pick $0<p<p_0$ such that $s<\frac{pd}{bd + \nu}$ and apply Theorem \ref{thm: Key Theorem} to such fixed $p$. By \eqref{eq-3g} in Lemma \ref{lemma-4a}, we have 
\[
\mathbb{E}\left\{ \Vert \langle X\rangle^{bd} f(H_\omega) \1_I(H_\omega) \delta_0\Vert^s_1 \right\} \leq C_f \sum_{k = 0}^\infty 2^{ksbd} 2^{-k(pd - s\nu)} = C \sum_{k = 0}^\infty 2^{-k(pd - s\nu - sbd)}<+\infty.
\]
This completes the proof of Theorem \ref{thm: SDL}. 
\end{proof}

\section{single-energy trap}
\label{sec: Preliminary_Lemmas}
In this section, we make some preparations and eventually derive the fixed energy trap, Proposition \ref{lemma: fixed_E0}, as discussed in \S \ref{subsec:sketch}. 

\subsection{Poisson formula}
Given $\Lambda\subset \ZZ^d$, let  
\[
    \partial\Lambda = \{(y,y')\in \ZZ^d\times \ZZ^d: |y - y'|_1 = 1, \text{~either~}y\in \Lambda, y'\notin \Lambda, \text{~or~}y'\in \Lambda, y\notin \Lambda.
\]
Assume $H_\omega\psi = E\psi$. Recall $H_{\omega, \Lambda} = P_\Lambda H_\omega P_\Lambda$, $G_{\omega, \Lambda, E} := (H_{\omega, \Lambda} - E)^{-1}$. Then the well-known Poisson's formula (c.f. \cite[(9.10)]{K07}) states that 
\begin{equation}
    \label{eq: Poisson}
    \psi(x) = -\sum\limits_{\substack{(y,y')\in \partial \Lambda\\ y\in \Lambda, y'\notin \Lambda}}G_{\omega, \Lambda,E}(x,y) \psi(y').
\end{equation}

\subsection{Stability of goodness}
The first lemma below describes the stability of ``goodness of boxes'' under exponential perturbation of energy $E_0$.
\begin{lemma}[Stability of goodness]\label{lemma: stability}
  Assume $\omega,E_0,L_0,x$ are fixed and $\forall L\geq L_0$, $\Lambda_L(x)$ is $(\omega,E_0,m_0,\eta_0)$-good. Then for any $m<m'<m_0$, there is $L_1 = L_1(m,m')$, s.t. $\forall L\geq L_1$, $\forall E $ satisfying $|E-E_0|\leq e^{-m'L}$, we have $\Lambda_L(x)$ is $(\omega,E,m,\eta_0)$-jgood.
\end{lemma}
\begin{proof}
  Recall the resolvent identity:
  \[
  G_{\omega,E,\Lambda_L(x)}-G_{\omega,E_0,\Lambda_L(x)}=(E_0-E)G_{\omega,E,\Lambda_L(x)}G_{\omega,E_0,\Lambda_L(x)}.
  \]
  Thus,
  \[\begin{aligned}
  \Vert G_{\omega,E,\Lambda_L(x)}\Vert  &\leq \Vert G_{\omega,E_0,\Lambda_L(x)}\Vert +|E_0-E| \cdot \Vert G_{\omega,E,\Lambda_L(x)}\Vert \cdot \Vert G_{\omega,E_0,\Lambda_L(x)}\Vert \\
    &\leq Ce^{L^{1-\eta}}+Ce^{-m'L+L^{1-\eta}}\Vert G_{\omega,E,\Lambda_L(x)}\Vert ,
  \end{aligned}
  \]
  and so,
  \[
    \Vert G_{\omega,E,\Lambda_L(x)}\Vert \leq \frac{Ce^{L^{1-\eta}}}{1-Ce^{-m'L+L^{1-\eta}}} \leq 2Ce^{L^{1-\eta}}.
  \]
  Also, if $m<m'$,
  \[
  \begin{aligned}
    |G_{\omega,E,\Lambda_L(x)}(a,b)|&\leq |G_{\omega,E_0,\Lambda_L(x)}(a,b)|+|E-E_0| \cdot |G_{\omega,E,\Lambda_L(x)}(a,b)| \cdot |G_{\omega,E_0,\Lambda_L(x)}(a,b)|\\
                                    &\leq e^{-m_0|a-b|}+e^{-m'L}e^{L^{1-\eta}+L^{1-\eta}}\\
                                    &\leq e^{-m|a-b|}
  \end{aligned}
  \]
  for $|a - b| \geq \frac{L}{100}$ when $L$ is large enough. Denote the threshold by $L_1$.
\end{proof}


\subsection{Coarse lattice}
To state the other lemmas, we need the following definitions.

\begin{definition}[Coarse lattice]\label{def-6}
    Fix $l>10$. Let $\alpha_l:= \floor{\tfrac{3l}{5}}$. Then we define $\mathcal C_l=(\alpha_l\ZZ)^d$ to be the coarse lattice.
\end{definition}
\begin{remark}\label{rmk-coarse}
Note that for any $10<l<L$, $\Lambda_l$ boxes centered at the coarse lattice $\mathcal C_l\cap \Lambda_L(x_0)$ cover the whole $\Lambda_L(x_0)$\footnote{\cite{GK12} call it the standard $l$-covering of $\Lambda_L(x_0)$.}, i.e. 
\[\Lambda_L(x_0) \subset \bigcup\limits_{x\in \mathcal C_l\cap \Lambda_L(x_0)} \Lambda_l(x) .\]
\end{remark}
We use coarse lattice when we want to use boxes of size $l$ to cover certain regions but do not want them to be too close to each other. 
\subsection{Predecessor of good box}
\begin{definition}\label{def-pgood}
    Given $0<r<1$, we say $\Lambda_L(x_0)$ is $(\omega, E, m_0, \eta_0)$-pgood (for predecessor of good) if, letting $l = L^{\frac{1}{1 + r}}$, every $\lambda_l(r)$, $r \in \mathcal C_l \cap \Lambda_L(x_0)$ is $(\omega, E, m_0, \eta_0)$-good. 
\end{definition}

\begin{lemma}[Probability for pgood box.]\label{lemma-5a}
    If scale $l$ is $(E, m_0, \eta_0, p_0)$-good, then for $L = l^{1 + r}$, $0<r<p_0$, we have for any $x_0$,  
    \[
    \PP\{\Lambda_L(x_0) \text{~is~} (\omega, E, m_0, \eta_0)-\text{pgood}\}\geq 1 - 2^d L^{-\frac{p_0 - r}{1 + r}d}.
    \]
\end{lemma}
\begin{proof}
    By the definition of coarse lattice, $\#(\mathcal C_l \cap \Lambda_L(x_0)) \leq \left(\frac{2L}{l}\right)^d$. By definition of a good scale, we get 
    \[
    \PP\{\Lambda_L \text{~is~not}-\text{pgood}\}\leq \left( \frac{2L}{l}\right)^d l^{-p_0d} = 2^d L^{-\frac{p_0 - r}{1 + r}d}.
    \]
\end{proof}

\begin{lemma}[Stability of pgood]\label{lemma-5b}
    Assume $l = L^{\frac{1}{1 + r}}$ and $\Lambda_L(x_0)$ is $(\omega, E_0, m_0, \eta_0)$-pgood for some $\omega$. For any $m<m'<m_0$, there is $L_2$, such that for any $L\geq L_2$, for any $|E - E_0|\leq e^{-m' l}$, $\Lambda_L(x_0)$ is $(\omega, E, m, \eta_0)$-jgood.  
\end{lemma}
\begin{remark}\label{rmk-5a}
    Note stability for good boxes, Lemma \ref{lemma: stability} said $\Lambda_L(x_0)$ remains $(\omega, E, m, \eta_0)$-good when $|E - E_0|\leq e^{-m_0L}$. While the stability for pgood set allows more perturbation: $|E - E_0|\leq e^{-m_0l}$ gives $(\omega, E, m, \eta_0)$-good $\Lambda_L$ boxes. Since $L = l^{1 + r}$, the stability of pgood set allows us to increase the scale (in the sense of ``multi-scale'' analysis).  
\end{remark}
\begin{proof}
    Since $\Lambda_L(x_0)$ is pgood, for any $r\in \mathcal C_{l} \cap \Lambda_L(x_0)$, $\Lambda_l(r)$ is $(\omega, E_0, m_0, \eta_0)$-good. By Lemma \ref{lemma: stability}, for any $|E - E_0|\leq e^{-m'l}$, $\Lambda_l(r)$ is $(\omega, E, m, \eta_0)$-good. 
    
    Recall the geometric resolvent identity (see for example \cite[Lemma 6.1]{DS18} for discrete version): For $x\in \Lambda'\subset \Lambda$, $y\in \Lambda$, we have 
    \[
    G_{\omega, E, \Lambda}(x, y) = G_{\omega, E, \Lambda'}(x, y) + \sum\limits_{\substack{u\in \Lambda'\\ v\in \Lambda\setminus \Lambda'\\ |u - v| = 1}}G_{\omega, E, \Lambda'}(x, u) G_{\omega, E, \Lambda}(v, y).
    \]
    In particular, for any $x\in \Lambda_L(x_0)$, there is a $(\omega, E, m, \eta_0)$-good box $\Lambda_l(r)$ covers $x$ and $d(x, \Lambda_l(r)^c)\geq l/10$\footnote{This is because the coarse lattice is size $\frac{3l}{5}$ while the boxes is size $l$. We just need to choose $r$ to be the closest coarse center from $x$, then $d(x, \Lambda_l(r)^c)\geq l/10$ is guaranteed.}. Applying the geometric resolvent identity to $\Lambda = \Lambda_L(x_0)$ and $\Lambda' = \Lambda_l(r)$, since $\Lambda_l(r)$ is good, 
    \begin{equation}
        \label{eq-5c}
        \begin{split}
    |G_{\omega, E, \Lambda_L(x_0)}(x, y)| &\leq \Vert G_{\omega, E, \Lambda_l(r)}\Vert + Cl^{d - 1} e^{-ml/10} \Vert G_{\omega, E, \Lambda_L(x_0)}\Vert\\
    &\leq e^{l^{1 - \eta_0}} + \frac{1}{2}\Vert G_{\omega, E, \Lambda_L(x_0)}\Vert. 
        \end{split}
    \end{equation}
    As a result, 
    \[
    \Vert G_{\omega, E, \Lambda_L(x_0)} \Vert \leq e^{l^{1 - \eta_0}} + \frac{1}{2}\Vert G_{\omega, E, \Lambda_L(x_0)} \Vert \quad \Rightarrow\quad \Vert G_{\omega, E, \Lambda_L(x_0)} \Vert \leq 2e^{l^{1 - \eta_0}}.
    \]
    Furthermore, when $|x - y|\geq L/100$, when $L$ is large enough, $y\notin \Lambda_l(r)$ that contains $x$, applying geometric resolvent identity, again, to $\Lambda = \Lambda_L(x_0)$ and $\Lambda' = \Lambda_l(r)$, and use $\Lambda_l(r)$ is good, we get 
    \[
    \begin{split}
    |G_{\omega, E, \Lambda_L(x_0)}(x, y)| &\leq 0 + Cl^{d - 1} e^{-\frac{ml}{10}} |G_{\omega, E, \Lambda_L(x_0)}(v, y)|\\
    &\leq Ce^{-ml/15}|G_{\omega, E, \Lambda_L(x_0)}(v, y)|   
    \end{split}
    \]
    for some $v$. Then we can repeat the whole process $\floor{\frac{15|x - y|}{l}} + 2$-many times, we will get 
    \[
    \begin{split}
    |G_{\omega, E, \Lambda_L(x_0)}(v, y)| &\leq e^{-m|x - y| - \frac{ml}{15}}|G_{\omega, E, \Lambda_L(x_0)}(z, y)|\\
    &\leq e^{-m|x - y| - \frac{ml}{15}} 2e^{l^{1 - \eta_0}}\leq e^{ -m|x - y|}.        
    \end{split}
    \]
    As a result, $\Lambda_L(x_0)$ is $(\omega, E, m_0, \eta_0)$-jgood.
\end{proof}

\subsection{Fixed energy trap}
The following lemma is mentioned in the idea of the proof of Theorem \ref{thm: Key Theorem} in Subsec. \ref{subsec: key_concept_key_theorem}. Under the same assumption of ``single-energy MSA result'', when $L$ is large enough, given some $E_0$, with high probability, one can make $\Wlxe$ exponentially small for any $|E - E_0|\leq e^{-mL}$.

\begin{prop}\label{lemma: fixed_E0}
    Under the same assumption with Theorem \ref{thm: Key Theorem}, for any $p'<p_0$, $m<m'<m_0$, when $L$ is large enough, for any $x_0\in \ZZ^d$, give any $E_0$, there is event $\mathcal R_{L,x_0}^{(E_0)}$ such that 
    \begin{enumerate}
        \item $\mathcal R_{L,x_0}^{(E_0)} \in \mathcal F_{\Lambda_L(x_0)}$ and $\PP(\mathcal R_{L,x_0}^{(E_0)}) \geq 1 - L^{-p' d}$. 
        \item for any $\omega\in \mathcal R_{L, x_0}^{(E_0)}$, $\Wlxe \leq e^{-\frac{m}{100}L}$ for any $|E - E_0| \leq e^{-\frac{m'}{100}L}$. 
    \end{enumerate}
\end{prop}

\begin{proof}
    By assumption, when $L$ is large enough, scale $l:=\frac{L}{100}$ is $(E_0, m_0, \eta_0, p_0)$-good. Consider the coarse lattice $\mathcal C_l$. Let $L_+ = L + \frac{L}{100}$, $L_- = L - \frac{L}{100}$. Set 
    \[
        \mathcal R_{L,x_0}^{(E_0)}:= \bigcap_{x\in \mathcal C_l \cap \Lambda_{2L_-,L_+}}\left\{ \omega: \Lambda_l(x) \text{~is~} (\omega, E_0, m_0,\eta_0)-\text{good}\right\}.
    \]
    Then \begin{enumerate}
        \item $\PP(\mathcal R_{L,x_0}^{(E)}) \geq 1 - (\frac{1000}{3})^d (\frac{L}{100})^{-p_0d}\geq 1 - L^{-p'd}$ when $L$ is large enough.
        \item If $\omega \in \mathcal R_{L,x_0}^{(E)}$, then all $\Lambda_l(x)$ is $(\omega, E_0,m_0,\eta_0)$-good. By Lemma \ref{lemma: stability}, for any $m<m''<m'<m_0$, for any $|E - E_0| \leq e^{-\frac{m'}{100}L}$, $\Lambda_l(x)$ is $(\omega, E, m'', \eta_0)$-jgood. Thus by \eqref{eq: Poisson}, for any $\psi_E\in \Theta_{\omega, E}$, 
        \[
        \begin{split}
        \Vert \psi_E\Vert_{\ell^2(\Lambda_{2L, L})} &\leq l^{d - 1}e^{-m''l} \sup_{x\in \Lambda_{2L_+,L_-}(x_0)}|\psi_E(x)| \\
        &\leq (\tfrac{L}{100})^{d - 1}e^{-\frac{m''}{100}L}(2L_+)^\nu\Vert T_{x_0}^{-1} \psi_E\Vert_{\ell^2}\\
        &\leq e^{-\frac{m}{100}L}\Vert T_{x_0}^{-1} \psi_E\Vert_{\ell^2}
        \end{split}
        \]
        when $L$ is large enough. Hence we obtain $\Wlxoe\leq e^{-\frac{m}{100}L}$.
    \end{enumerate}
\end{proof}

\section{Site percolation}\label{sec-percolation}
In this section, develop the percolation argument, Corollary \ref{cor-6b}, \ref{cor-7a} that will be needed for the spectral reduction in the next section. 
\begin{definition}[Good nodes and loops]\label{def-goodloop}
Recall $\mathcal C_l$ denote the coarse lattice, see Definition \ref{def-6}. Fix $\omega$, we say
   \begin{enumerate}
     \item $x\in\mathcal{N}_l$ is a $(\omega,E_0,m_0,\eta_0)$-good (-bad) node if $\Lambda_l(x)$ is a $(\omega,E_0,m_0,\eta_0)$-good (-bad) box.
     \vspace{0.05in}
     \item $A\subset\mathcal{N}_l$ is a $(l,E_0,m_0,\eta_0)$-good shell if each node in $A$ is a good node and $\mathcal A$ is a finite set such that $\mathcal C_l\setminus A = B \bigsqcup C$ with $d(B, C)>3l$, i.e. $\mathcal A$ splits $\mathcal C_l$ into two parts. We denote the finite one among $B$ and $C$ by $A_{in}$; the other one by $A_{out}$.  
     \vspace{0.05in}
     \item We say a $(l,E_0,m_0,\eta_0)$-good shell $\mathcal{A}$ is fully contained in  $S\subset\ZZ^d$ if $\bigcup\limits_{x\in\mathcal{A}}\Lambda_{l+2}(x)\subset S$.
   \end{enumerate}
\end{definition}
\begin{lemma}[Good shell]\label{lemma-7a}
  Let $l>12$. For fixed $\omega$, if there is a $(l,E_0,m_0,\eta_0)$-good shell $\mathcal{A}$  that is fully contained in $\Lambda_{L_2,L_1}(x_0)$, then for any $m<m'<m_0$, $|E-E_0|\leq e^{-m'l}$, $E\in \mathcal I$, we have
  \begin{equation}
      \label{eq-6b}
  \text{dist}(E,\sigma^{(\mathcal I)}(H_{\Lambda_{L_2}})) W_\omega(x_0;E)\leq L_2^{2\nu}e^{-\frac{m l}{3}}
  \end{equation}
    where $\sigma^{(\mathcal{I})}(H) = \sigma(H)\cap \mathcal I$. In particular, if $l = \sqrt{L}$, $L_1 = \frac{L}{2}$, $L_2 = L$, then when $L$ is large enough, 
  \begin{equation}
      \label{eq-6a}
    \text{If~}\Wxe \geq e^{-\frac{m}{30}\sqrt{L}} \quad \Rightarrow \quad  dist(E, \sigma(H_{\omega, \Lambda_L(x_0)})\leq e^{-\frac{m}{30}\sqrt{L}}.
  \end{equation}
\end{lemma}
\begin{proof}
Note that
 \[
   \begin{split}
     \text{dist}(E,\sigma(H_{\Lambda})) = \Vert (H_{\Lambda}-E)^{-1}\Vert ^{-1}= \inf_{\psi\neq 0 } \frac{\Vert (H_{\Lambda}-E)\psi\Vert }{\Vert \psi\Vert }.
   \end{split}
 \]
 Let $\phi_E$ be a generalized eigenfunction of $H_\omega$ with respect to  $E$. Then $(H_\omega - E)\phi_E = 0$. Take $\psi = \1_{A_{in}}$. Then $((H_{\omega, \Lambda_L(x_0)} - E)\phi_E\psi)(x) = 0$ when $d(x, A) >2$. When $d(x, A) \leq 2$, there is $r\in A$ such that $\Lambda_l(r)$ is a $(\omega, E_0, m_0, \eta_0)$-good box that contains $x$. Since $|E - E_0|\leq e^{-m'l}$, by Lemma \ref{lemma: stability}, $\Lambda_l(r)$ is $(\omega, E, m, \eta_0)$-good. We can use Poisson's formula on those $\phi_E(x)$ to get  
\[
\begin{split}
    \Vert (H_{\omega, \Lambda_L(x_0)} - E)\phi_E\psi\Vert^2 &= \sum_{d(x, A)\leq 2}  |((H_{\omega, \Lambda_L(x_0)} - E)\phi_E\psi)(x)|^2\\
    &\leq L_2^d \Vert H_{\omega, \Lambda_L(x_0)} - E\Vert^2 \max_{d(x, A)\leq 2} |\phi_E(x)|^2\\
    &\leq CL_2^d l^{d - 1}e^{-m(l - 2)}\max_{x\in \Lambda_{L_2,L_1}(x_0)}|\phi_E(x)|^2
\end{split}
\]
where we also used $A$ is fully contained in $\Lambda_{L_2, L_1}(x_0)$. Note that 
\[
\max\limits_{ y\in\Lambda_{L_2,L_1}(x_0) }|\phi_E(y)| = \max\limits_{ y\in\Lambda_{L_2,L_1}(x_0) }\langle y-x_0\rangle^\nu \frac{|\phi_E(y)|}{\langle y-x_0\rangle^\nu}\leq \langle L_2\rangle^\nu \Vert T_{x_0}^{-1}\phi_E\Vert.
\]
As a result, 
\[
\frac{\Vert (H_{\omega, \Lambda_L(x_0)} - E)\phi_E\psi\Vert}{\Vert \phi_E\psi\Vert}\leq L_2^{\frac{d}{2}} e^{-\frac{ml}{3}}\langle L_2\rangle^\nu \frac{\Vert T_{x_0}^{-1}\phi\Vert}{|\phi_E(x_0)|}
\]
when $l$ is large enough. Hence by definition of $\Wxe$, 
\[
d(E, \sigma(H_{\omega, \Lambda_L(x_0)})) \leq L_2^{\frac{d}{2} + \nu}e^{-\frac{ml}{3}} \Wxoe^{-1} \leq L_2^{2\nu} e^{-\frac{ml}{3}} \Wxoe^{-1}.
\]
This completes the proof of \eqref{eq-6b}. In particular, when $l = \sqrt{L}$, $L_1 = \frac{L}{2}$, $L_2 = L$, when $L$ is large enough, we have 
\[
\Wxoe \dist(E, \sigma^{(\mathcal I)}(H_{\Lambda_{L_2}}) \leq L^{2\nu}e^{-\frac{m\sqrt{L}}{3}}.
\]
Hence we obtain \eqref{eq-6a}.
\end{proof}

\begin{definition}\label{def-Yset}
    Let $\mathcal{Y}^{(E_0)}_{x_0,l,L_1,L_2}\in \mathcal F_{\Lambda_{L_2, L_1}(x_0)}$ denote the event
\[\{\omega:\text{there~is~an~}(\omega,l,E_0,m_0,\eta_0)\text{-good~loop~fully~contained~in~}\Lambda_{L_2,L_1}(x_0)\}.\]
\end{definition}
If a scale $l$ is good, each node has a large probability of being ``good'', and we expect a relatively large probability for having good loops as well. The next Lemma quantifies this intuition.



\begin{lemma}[Probability for good shell]\label{lemma-7b}
 Assume $E_0$ is fixed, and the scale $l$ is $(E_0,m_0,\eta_0,p_0)$-good. We have
 \begin{equation}\label{eq-5b}
   \PP\left\{\mathcal{Y}^{(E_0)}_{x_0,l,L_1,L_2}\right\}\geq 1-C\left(\frac{L_1+3l}{l}\right)^{d-1}(2^d)^{\frac{L_2-L_1-l}{l}}l^{-pd\frac{L_2-L_1-l}{(3^d-1)l}}.
 \end{equation}
 In particular, if $l=\sqrt{L}$, $L_1=\frac{L}{2}$, $L_2=L$, when $L$ is large enough, then
 \begin{equation}\label{eq-5a}
   \PP\left\{\mathcal{Y}^{(E_0)}_{x_0,\sqrt L,\frac{L}{2},L}\right\}\geq 1-L^{-c_{d,p}\sqrt{L}}.
 \end{equation}
\end{lemma}

\begin{proof}
  Fix $E_0\in\mathcal{I}$. First note that
  \[
  \begin{split}
  (\mathcal{Y}^{(E_0)}_{x_0,l,L_1,L_2})^c &= \{\omega: \text{there~is~no~good~shell~totally~inside~}\Lambda_{L_2,L_1}\}\\
  &= \{\omega:\text{there~is~a~bad~``path''\footnote{If we view coarse lattice as a graph, with edges that connect two nearest vertices, ``path'' refers to a path in this graph. }~}\partial\Lambda^+_{L_1+l+2}\text{~to~}\partial\Lambda^-_{L_2-l-2}\}
\end{split}
  \]
    Notice that each such bad ``path'' must contain at least $N:=\frac{L_2-L_1-2l-4}{6l/5}+1$ many bad nodes starting from $\partial\Lambda^+_{L_1+l+2}$, which means it should contain $\frac{N}{(3^d-1)l}$-many independent bad nodes. Each nodes is bad with probability $l^{-pd}$ by definition of good scale. And the number of all such potential paths is less than $2d(\frac{L_1+l+2}{3l/5})^{d-1}(2^d)^N$. Putting these observations together, we obtain
\[
  \begin{split}
    \PP\{(\mathcal{Y}^{(E_0)}_{x_0,l,\frac{L}{2},\frac{L}{2}})^c\} &\leq 2d(\frac{L_1+l+2}{3l/5}+1)^{d-1}(2^d)^N \cdot l^{-pdN}\\
    &\leq  C(\frac{L_1+3l}{l})^{d-1}(2^d)^{\frac{L_2-L_1-l}{l}} \cdot l^{-pd\left(\frac{L_2-L_1-l}{(3^d-1)l}\right)}
  \end{split}
\]
This gives us \eqref{eq-5b}. By taking $l = \sqrt{L}$ and letting $L$ be sufficiently large, we get \eqref{eq-5a}.
\end{proof}

\begin{cor}\label{cor-7a}
    Assume $E_0\in \RR$, $l< L_1 < L_2$ are fixed and scale $l$ is $(E_0, m_0, \eta_0, p_0)$-good. Then for any $x_0$, there is an event $\mathcal Y_{x_0, l, L_1, L_2}^{(E_0)}\in \mathcal F_{\Lambda_{L_2, L_1}(x_0)}$ such that 
    \[
    \PP\left\{\mathcal{Y}^{(E_0)}_{x_0,l,L_1,L_2}\right\}\geq 1-2d\left(\frac{L_1+3l}{l}\right)^{d-1}(2^d)^{\frac{L_2-L_1-l}{l}}l^{-pd\frac{L_2-L_1-l}{(3^d-1)l}}
    \]
    and for any $\omega\in \mathcal{Y}^{(E)}_{x_0,l,L_1,L_2}$,  $m<m'<m_0$, $|E - E_0|\leq e^{-m'l}$, there is  $x_0\in \ZZ^d$, such that 
    \[
    \dist(E, \sigma^{(\mathcal I)}(H_{\Lambda_{L_2}})) \Wxoe \leq L_2^{2\nu} e^{-\frac{ml}{3}}.
    \]
\end{cor}
\begin{proof}
    This follows from a combination of Lemma \ref{lemma-7a} and \ref{lemma-7b}.
\end{proof}
    The percolation method naturally generalizes from $\Lambda_l$ good boxes to $\Lambda_L$-pgood boxes (with $L = l^{1 + r}$, $r<p_0$) without much effort. 
\begin{cor}\label{cor-6b}
    Assume scale $l$ is $(E_0, m_0, \eta_0, p_0)$ good and $L = l^{1 + r}$. There is a set $\mathcal P_{x_0, L, L_1, L_2}^{(E_0)}\in \mathcal F_{\Lambda_{L_2, L_1}(x_0)}$ such that 
    \[
    \PP\left\{\mathcal{P}^{(E_0)}_{x_0,L,L_1,L_2}\right\}\geq 1-2d\left(\frac{L_1+3L}{L}\right)^{d-1}(2^d)^{\frac{L_2-L_1-L}{L}}L^{-\hat{p}d\frac{L_2-L_1-L}{(3^d-1)L}}
    \]
    and for any $\omega\in \mathcal{P}^{(E_0)}_{x_0,L,L_1,L_2}$, $m<m'<m_0$, $|E - E_0|\leq e^{-m_0 l}$, we have 
    \[
    \dist(E, \sigma(H_{\Lambda_{L_2}})) \Wxoe \leq L_2^{2\nu} e^{-\frac{mL}{3}}.
    \]
\end{cor}
\begin{proof}
    Let $\mathcal P_{x_0, L, L_1, L_2}^{(E_0)}\in \mathcal F_{\Lambda_{L_2, L_1}(x_0)}$ denote the event 
    \[
    \{\omega:\text{there~is~an~}(\omega,L,E_0,m,s)\text{-pgood~loop~fully~contained~in~}\Lambda_{L_2,L_1}(x_0)\}
    \]
    where ``pgood-loop'' refers to $\mathcal C_L$ loop where each site has a $\Lambda_L$ pgood box (comparing to Definition \ref{def-goodloop}, \ref{def-Yset}). Since scale $l$ is $(E_0,  m_0, \eta_0, p_0)$-good, by Lemma \ref{lemma-5a}, 
    \[
    \PP\{\Lambda_L(x_0)\text{~is~}(\omega, E_0, m_0, \eta_0)\text{-pgood}\}\geq 1 - CL^{-\hat{p}d}, \quad \hat{p} = \frac{p_0 - r}{1 + r}.
    \]
    Thus one just needs to replace the ``good $\Lambda_l$ boxes'' with ``pgood $\Lambda_L$ boxes'', replace $p_0$ by $\hat{p}$, replace Lemma \ref{lemma: stability} by Lemma \ref{lemma-5b}, the proofs of Lemma \ref{lemma-7a}, \ref{lemma-7b} work directly. Notice that under the same level of perturbation $|E - E_0|\leq e^{-m'l}$, Corollary \ref{cor-6b} derived better result $e^{-mL/3}$ compared to Corollary \ref{cor-7a}, which is $e^{-ml/3}$. This is due to the fact that stability for a pgood set is stronger than stability for a good set, see Remark \ref{rmk-5a}.
\end{proof}

\section{Proof of key theorem}
\label{sec: two_spectral_reduction}
We prove Theorem \ref{thm: Key Theorem} in this section by performing two spectral reductions: Theorem \ref{side_1} and Theorem \ref{side_2}. Before that, we set up several constants that will be used in the proof of first and second spectral reduction: Let 
\begin{equation}
    \label{eq-7g}
N_1:= \min \{n\in \NN: 2^{\frac{1}{n}} - 1<p_0\}, \quad M = \frac{m_0}{30^{N_1 + 2}}
\end{equation}
and set $0<\rho, \beta<1$ and $N_2\in \NN$ to be such that  
\begin{equation}
    \label{eq-7d}
    (1 + p_0)^{-1} <\rho <1, \quad \beta = \rho^{N_2}, \quad (N_2 + 1)\beta <p_0 - p.
\end{equation}
Assume we are under the same assumptions as Theorem \ref{thm: Key Theorem} in the rest of this section.
\subsection{The first spectral reduction}\label{subsec: first_reduction}

\begin{thm}\label{side_1}
  Given any $b\geq 1$, there exists a constant $K_{d,p,b}\geq 1$ s.t. for any $K\geq K_{d,p,b}$, for large enough $L$, for any $x_0\in\ZZ^d$, there is an event $\mathcal{O}_{L, x_0}\in\mathcal{F}_{\Lambda_L(x_0)}$ such that 
  \[
  \PP\{\mathcal{O}_{L, x_0}\}\geq 1-L^{-2bd}. 
  \]
  and for any $\omega\in\mathcal{O}_{L, x_0}$, if 
  \begin{equation}
  \label{eq-7b}
  \Wxoe\geq e^{-30M\sqrt{L/K}}, \quad \dist(E, \mathcal I^c) \geq e^{-30M\sqrt{L/K}},
  \end{equation}
  then 
  \begin{equation}
      \label{eq-7a}
  \dist(E, \sigma^{(\mathcal I)}(H_{\omega, \Lambda_L}))\leq e^{-30ML/K}.
  \end{equation}
  where $\sigma^{(\mathcal{I})}(H) = \sigma(H)\cap \mathcal I$.
\end{thm}

\begin{proof}
  The strategy here is two-fold:
  \begin{enumerate}
    \item Construct $\mathcal{O}_{L, x_0}$ by layers.
    \item Estimate the probability of the event $\mathcal{O}_{L, x_0}$.
  \end{enumerate}
  We first make some preparations:  Let $r>0$ be such that $(1 + r)^{N_1} = 2$. Given $L_0\in \RR$, we define 
  \[
  l_0 = \sqrt{L_0}, \quad l_k=l_{k-1}^{1+\eta} = (\sqrt{L_0})^{(1 + r)^k}, \quad \text{~for~} k=1,2,\cdots,N_1.
  \]
  In particular, $l_{N_1}=l_0^2=L_0$. Furthermore, let \[
  L_k=L_{k-1}+2Jl_{k}, \quad k = 1, \cdots, N_1, 
  \]
  where $J$ is a large constant that will be determined later. In particular, 
  \begin{equation}
      \label{eq-7c}
  L_{N_1}=L_0+2J\sum_{k=1}^{ N_1 }l_k\leq (1+2J N_1 )L_0 =:KL_0.
  \end{equation}
  Now we inductively construct $\mathcal O_{L, x_0}$: Given $L$ large enough, we find $L_0$ such that $L_{N_1}$ defined above equals $L$\footnote{The expression for $L_{N_1}$ depends continuously on $L_0$ so there must be a $L_0$ satisfying $L = L_{N_1}$.}.
\begin{enumerate}
  \item For the initial layer $\Lambda_{L_0}$, we pick $\mathcal{Y}^{E_{0,i}}_{x_0,l_0,\sqrt{L_0},L_0}\in \mathcal F_{\Lambda_L(x_0)}$ where $E_{0,i}$ are energies such that the union of $[E_{0,i}-e^{-m_0l_0},E_{0,i}+e^{-m_0l_0}]$ covers $\mathcal{I}$, indexed by $i$. We need to choose $\frac{|\mathcal{I}|}{2e^{-m_0l_0}}=O(e^{m_0\sqrt{L_0}})$ many of them.
  Let $\mathcal Y_0=\bigcap\limits_i \mathcal{Y}^{E_{0,i}}_{l_0,\sqrt{L_0},L_0}$, then there is $C, c$ such that 
  \[
    \PP\{\mathcal Y_0\}\geq 1-Ce^{m_0\sqrt{L_0}}L_0^{-c_{d,p}\sqrt{L_0}}\geq 1-CL_0^{-c\sqrt{L_0}}.
  \]
  \item Given $k = 1, \cdots N_1$, we define
  \[
  \mathcal Y_k := \bigcap_{E\in \sigma^{(\mathcal I)}(H_{\omega, \Lambda_{L_{k - 1}(x_0)}})} \mathcal Y_{l_k, L_{k - 1}, L_k}^{(E)} \in \mathcal F_{\Lambda_{L_k}(x_0)}\footnote{This is well-defined because $E$ only depends on $\mathcal F_{\Lambda_{L_{k - 1}}(x_0)}$ while $\mathcal Y_{l_k, L_{k - 1}, L_k}^{(E)} \in \mathcal F_{\Lambda_{L_k, L_{k - 1}}(x_0)}$. } 
  \]
  and 
  \[
  \mathcal O_{L, x_0} = \bigcap_{k = 0}^{N_1} \mathcal Y_k \in \mathcal F_{\Lambda_L(x_0)}.
  \]
\end{enumerate}
  By Lemma \ref{lemma-7b},  $\#(\sigma(H_{\omega, \Lambda_L})) = L^d$, we have 
  \[
  \begin{split}
  \PP(\mathcal Y_{l_k, L_{k - 1}, L_k}^{(E)}) &\geq 1 - C L_{k - 1}^d 2^{(2J - 1)d} l_k^{-p_0 d(2J - 1)/3^d} \geq 1 - CL_{N_1}^d 2^{2Jd} \cdot l_0^{-p_0 Jd/3^d}\\
  &\geq 1 - CL^d 2^{2Jd} L_0^{-p_0 Jd/3^{d +1}} \geq  1 - CL_0^{-5bd}
  \end{split}
  \]
  when we pick $J \geq J_0 :=  3^{d + 3} b$ and large enough $L_0$. Hence 
  \[
  \PP(\mathcal Y_k) \geq 1 - L_0^d L_0^{-5bd} \geq 1 - L_0^{-4bd},
  \]
  thus 
  \[
  \PP(\mathcal O_{L, x_0}) \geq 1 - N_1 L_0^{-4bd} \geq 1 - \left(\tfrac{L}{K}\right)^{-3bd}\geq 1 - L^{-2bd}, 
  \]
  when $L$ is large enough. Therefore we obtain the probability estimate as long as $J\geq J_0$, i.e. $K = 2N_1J + 1\geq K_{d,p,b}:= 2N_1 J_0 + 1$. It remains to prove \eqref{eq-7a}.
    First note that 
    \[
        \Wxoe \geq e^{-30M \sqrt{L/K}} \geq e^{-30M \sqrt{L_0}}, \qquad 30M = \frac{m_0}{30^{N_1 + 1}}.
    \]
    For any $E\in \mathcal I$, there is some $E_{0, i}$ such that $|E - E_{0, i}|\leq e^{-m_0 l_0}$. Since $\omega\in \mathcal Y_{l_0, \sqrt{L_0}, L_0}^{(E_{0,i})}$, by Corollary \ref{cor-6b}, this implies 
    \[
    \Wxoe \dist(E, \sigma^{(\mathcal I)}(H_{\omega, \Lambda_{L_0}(x_0)})) \leq e^{-\frac{m_0}{3}l_1}. 
    \]
    Since $\Wxoe\geq e^{-\frac{m_0}{30}l_1}$, thus 
    \[
    \dist(E, \sigma^{(\mathcal I)}(H_{\omega, \Lambda_{L_0}(x_0)})) \leq e^{ \left(-\frac{m_0}{3} + \frac{m_0}{30}\right)l_1}\leq e^{-\frac{m_0}{30}l_1} =: e^{-m_1 l_1}.
    \]
    Now there is some $E_{1, i}\in \sigma^{(\mathcal I)}(H_{\omega, \Lambda_{L_0}(x_0)})$ such that $|E - E_{1, i}|\leq e^{-m_1l_1}$. Since $\omega\in \mathcal Y_{l_1, L_0, L_1}^{(E_{1,i})}$, we can invoke Corollary \ref{cor-6b} again and repeat this process for $N_1$ times and we obtain 
  \[
    \text{dist}\left(E,\sigma^{(\mathcal{I})}(H_{\omega,\Lambda_{L_{N_1}}(x_0)})\right)\leq e^{-\frac{m_0}{30^{N_1 + 1}}l_{N_1 + 1}} \leq e^{-30 M \frac{L}{K}}.
  \]
\end{proof}

\subsection{Second spectral reduction}\label{subsec: second_reduction}
Recall that  $\rho<1$, $\beta<1$, $N_2$ were defined in \eqref{eq-7d}
\[
    (1 + p_0)^{-1} <\rho <1, \quad \beta = \rho^{N_2}, \quad (N_2 + 1)\beta <p_0 - p.
\]
In this subsection, given $L$, we let $L_0 = L$, $L_n = L^{\rho^n}$, $n = 1, \cdots, N_2$. In particular, $L_{N_2} = L^\beta$. 
\begin{definition}[reduced spectrum]\label{def-7a}
  The reduced spectrum of $H_\omega$ in $\Lambda_L(x_0)$, in the energy interval $\mathcal{I}$ is defined as
  \[
  \begin{split}
  &\sigma^{(\mathcal{I},red)}(H_{\omega,\Lambda_L(x_0)}):=\\
  &\left\{E\in\sigma^{(\mathcal{I})}(H_{\omega,\Lambda_L(x_0)}):\text{~dist~}\left(E,\sigma^{(\mathcal{I})}(H_{\omega,\Lambda_{L_n}(x_0)}\right)\leq 2e^{-\frac{30M}{K}L_n},n=1,\cdots,N_2\right\}
\end{split}
  \]
where $K$ is the constant given in Theorem \ref{side_1}.
\end{definition}

\begin{thm}[second spectral reduction]\label{side_2}
  Given any $b\geq 1$. Let $K\geq K_{d, p, b}$ be a constant as in Theorem \ref{side_1}. When $L$ is large enough, for each $x_0\in\ZZ^d$ there exists an event $\mathcal S _{L,x_0}\in \mathcal{F}_{\Lambda_L(x_0)}$, with
  \[
  \PP\{\mathcal S _{L,x_0}\}\geq 1-L^{-b\beta d},
  \]
  and for any $\omega\in\mathcal S_{L,x_0}$
  \begin{enumerate}
    \item If $E\in\mathcal{I}$ satisfies
    \begin{equation}\label{condition_2}
    W_\omega(x_0;E)>e^{-30M\sqrt{\frac{L^\beta}{K}}}\quad\text{and}\quad \text{dist~}(E,\mathcal{I}^c)>2^{-30M\sqrt{\frac{L^\beta}{K}}}
  \end{equation}
    then
    \begin{equation}\label{result_2}
    \text{dist}(E,\sigma^{(\mathcal{I},red)}(H_{\omega,\Lambda_L(x_0)}))\leq e^{-\frac{30M}{K}L}
  \end{equation}
    \item and we have
    \begin{equation}\label{result_3}
    \#\sigma^{(\mathcal{I},red)}(H_{\omega,\Lambda_L(x_0)})\leq CL^{( N_2 +1)\beta d}
  \end{equation}
  \end{enumerate}
\end{thm}


To obtain \eqref{result_2} from \eqref{condition_2}, one only needs to consider the event $\widetilde {\mathcal O}_{L, x_0}=\bigcap\limits_{n=0}^{N_2}\mathcal{O}_{L_n, x_0}$. By \eqref{condition_2}, $\Wxoe \geq e^{-30M \sqrt{L_n}}$ for each $n = 1, \cdots, N_2$. By Theorem \ref{side_1}, $\dist(E, \sigma^{(\mathcal I)}(H_{\omega, \Lambda_{L_n}(x_0)}))\leq e^{-30M \frac{L_n}{K}}$ for each $n$. Then \eqref{result_2} follows from Definition \ref{def-7a}. 

Therefore it remains to find an event such that \eqref{result_3} holds. We first make some preparations and then present the result in Lemma \ref{lemma: number of reduced spectrum}. 

\medskip  


Given\footnote{This $L$ is an arbitrary $L$. It is not necessarily the one in Theorem \ref{side_2}} $L'<L$ with $L^\rho<\frac{L - L'}{7}$, let $\rho$, $N_2$ be the one in \eqref{eq-7d} and $L_n=L^{\rho^n}$ for $n=0,1,2,\cdots, N_2 $. Let 
\[
\begin{split}
    R_n := \mathcal C_{L_n} \cap \Lambda_L(x_0),\qquad 
    \mathcal{R}_n := \{\Lambda_{L_n}(r)\}_{r\in R_n}.
\end{split}
\]
Recall that $\Lambda_l$ boxes centered at coarse lattices $\mathcal C_l$ cover the whole space. In particular, $\mathcal R_{n}$ covers $\Lambda_L(x_0)$ for any $n$. 

Given $K_2\in \mathbb{N}$ (where $K_2$ will be chosen later), we define
\begin{definition}
The annulus $\Lambda_{L, L'}(x_0)$ is $(\omega,E,K_2)$-notsobad if there are at most $K_2$ points in $R_{ N_2 }$, denoted by $r_i, 1\leq i\leq K_2$, s.t. $\forall x\in \Lambda_{L, L'}(x_0)\setminus \left(\bigcup_{r_i}\Lambda_{3L_{ N_2 }}(r_i)\right)$,
there exists some $n\in\{1,2,\dots, N_2\}$, such that  $\Lambda_{L_{n}}(r)\in \mathcal R_n$ is a $(\omega,E,m_0,\eta_0)$-good box and $\Lambda_{\frac{L_n}{5}}(x) \cap \Lambda_{L, L'} \subset \Lambda_{L_n}(r)$. 

An event $\mathcal{N}$ is $(\Lambda_{L, L'}(x_0),E,K_2)$-notsobad if $\mathcal{N}\in\mathcal{F}_{\Lambda_{L, L'}(x_0)}$, and $\Lambda_{L, L'}(x_0)$ is $(\omega,E,K_2)$-notsobad for all $\omega\in\mathcal{N}$.
\end{definition}
\begin{remark}\label{rmk-singularset}
    $\Theta:= \bigcup_{r_i}\Lambda_{3L_{ N_2 }}(r_i)$ is called the singular set and the above definition captures the fact that outside of the singular set, each point is good in at least one level $L_n$, $n\in\{1,2,\dots, N_2 \}$.
\end{remark} 

\begin{lemma}\label{side_2_2}
  Assume $K_2$ and $L$ are large enough, $L^\rho< \frac{L - L'}{7}$. Given a fixed $E\in\mathcal{I}$, there exists a $(\Lambda_{L, L'}(x_0),E,K_2)$-notsobad event $\mathcal{N}_{\Lambda_{L, L'}(x_0)}^{(E)}$ with
  \[
  \PP\{\mathcal{N}_{\Lambda_{L, L'}(x_0)}^{(E)}\}>1-L^{-5bd}
  \]
\end{lemma}
\begin{proof}
    Given $\Lambda_{L_{n - 1}}(r) \in \mathcal R_{n - 1}$, we set 
    \[
    \begin{split}
    &\mathcal R_n(r) := \{ \Lambda_{L_n}(s): \Lambda_{L_n}(s) \cap  \Lambda_{L_{n - 1}}(r) \neq \emptyset\}, \\
    &R_n(r) := \{s\in R_n: \Lambda_{L_n}(s) \in \mathcal R_n(r)\}. 
    \end{split}
    \]
    By definition, $\Lambda_{L_{n - 1}}(r)\subset \bigcup\limits_{s\in R_n(r)}\Lambda_{L_n}(s)$.  and by the definition of coarse lattice, $\# R_n(r) \leq \left(\frac{3L_{n - 1}}{L_n}\right)^d$. Let $N_{n -1}(r)$ denote the number of bad boxes among $\Lambda_{L_n}(s)$. Let
    \[
    \mathcal N_{\Lambda_{L, L'}}^{(E)}:= \bigcap_{n = 1}^{N_2} \bigcap_{r\in R_{n - 1}} \bigcap_{k = 1}^{K_2} \mathcal N_{n,r,k}
    \]
    where $\mathcal N_{n,r,k}$ denote the set of $\omega$ where all $\Lambda_{L_n}(s)$ boxes with $s\in R_n(r)$ except at most $K_2$ many disjoint ones, are $(\omega, E, m_0, \eta_0)-$good. It is clear by the definition that it is a $(\Lambda_{L, L'}, E, K_2)$-notsobad set. Furthermore, we can estimate the probability of the complement set 
    \[
    \begin{split}
        \PP\left\{\left(\mathcal N_{\Lambda_{L, L'}}^{(E)}\right)^c\right\} &\leq \sum_{n = 1}^{N_2} \left(\frac{3(L - L')}{L_{n - 1}}\right)^d \left( \frac{3L_{n - 1}}{L_{n}}\right)^{K_2d} L_n^{-K_2p_0 d}\\
        &\leq 2^d3^{K_2d} N_2 L^{-\rho^{N_2 - 1}(K_2+ (\rho(p_0d + d) - d) + d) + d} \leq L^{-5b d}
    \end{split}
    \]
    where we choose $K_2$ large enough but fixed and then $L$ large enough for the last inequality. 
\end{proof}
Let
\[
\mathcal{N}_{\Lambda_{L, L'}(x_0)}=\bigcap_{E\in\sigma^{(\mathcal I)}(H_{\omega,\Lambda_{L'}})} \mathcal{N}_{\Lambda_{L, L'}(x_0)}^{(E)}\in\mathcal{F}_{\Lambda_L}.
\]
By Lemma \ref{side_2_2}, $\PP(\mathcal N_{\Lambda_{L, L'}})\geq 1 - L^{-4b d}$ when $L$ is large enough. 

\begin{lemma}\label{lemma: number of reduced spectrum}
    Given $b\geq 1$, when $L$ is large enough, for any $x_0\in \ZZ^d$, there is $\mathcal N_{L, x_0}$ such that 
    \[
    \PP\{\mathcal N_{L, x_0}\} > 1 - CL^{-4b\beta d}
    \]
    and for any $\omega \in\mathcal N_{L, x_0}$, 
    \begin{equation}
        \label{eq-7e}
    \# \sigma^{(\mathcal I, \red)}(H_{\omega, \Lambda_L(x_0)})\leq C L^{(N_2 + 1)\beta d}. 
    \end{equation}
\end{lemma}
\begin{proof}
Recall that $L_0 = L$, $L_n = L^{\rho^n}$, $L_{N_2} = L^\beta$. Define 
    \[
    \mathcal{N}_{L,x_0}=\bigcap_{n=0}^{ N_2 }\mathcal{N}_{\Lambda_{L_{n-1},L_n(x_0)}}.
    \]
    Then $\PP(\mathcal N_{L, x_0}) \geq 1 - N_2 L_{N_2}^{-4b d} \geq 1 - L^{-3b\beta d}$ when $L$ is large enough. It remains to show \eqref{eq-7e}.  Denote 
 \[
 D_k^{N_2}:= \left\{ \{E_n\}_{n = k}^{N_2}: E_n\in \sigma(H_{\omega, \Lambda_{L_n}}),~ |E_i - E_j|\leq 2e^{-\frac{30M}{K}L_{\max\{i,j\}}}\right\}. 
 \]
 Then by definition, 
\begin{equation}
  \begin{split}
  \#\sigma^{(\mathcal{I},red)}(H_{\omega,\Lambda_{L}}) \leq  \# D_0^{ N_2 }
\end{split}
\end{equation}
  We can count the RHS by layers inductively. We start with the layer $L_{ N_2 }$ and we omit $x_0$ and $\omega$ below for convenience. We first have 
  \[
  \#D_{ N_2 }^{ N_2 }=\#\sigma^{(\mathcal I)}(H_{\omega,\Lambda_{L_{ N_2 }}})\leq C(L_{ N_2 })^d = CL^{\beta d}.
  \]
  Given $\bfek := \{E_n\}_{k}^{ N_2 }\in D_k^{ N_2 }$, we compute the number of $E$ in 
  \[
  B_{k - 1}(\bfek):= \{E:\text{if~} E_{k-1}=E, \text{~then~}\{E_n\}_{k-1}^{ N_2 }\in D_{k-1}^{ N_2 }\}.
  \]
  Then we get the recursion relation
  \begin{equation}
      \label{eq-7f}
  \# D_{k - 1}^{N_2} \leq \# D_{k}^{N_2} \times \left(\max_{\bfek \in D_k^{N_2}} \# B_{k - 1}(\bfek)\right).
  \end{equation}
  Since $\omega\in\mathcal{N}_{\Lambda_{L_{n-1},L_n}}$ for any n, we see that $\Lambda_{L_{n-1},L_n}$ is an $(\omega,L_{n-1},L_n,E_n)$-notsobad set. Let $\Theta_n$ be the corresponding singular set, see Remark \ref{rmk-singularset}. Notice $\Theta_n\subset \Lambda_{L_{n - 1}, L_n}$, and set $\Theta_k^{ N_2 }=\bigcup_{n=k}^{ N_2 }\Theta_n\cup\Lambda_{L_{ N_2 }}$. Hence we have
  \[
  |\Theta_k^{ N_2 }|\leq L_{ N_2 }^d+\sum_{n=k}^{ N_2 }K_2(3(L_{k-1})_{ N_2 })^d=L^{\beta d}+( N_2 -k+1)K_23^dL^{\rho^{k-1}\beta d}\leq CL^{\beta d}.
  \]
  Now if $E\in B_{k-1}(\bfek)$, then for any $ x \in\Lambda_{L_{k-1}}\setminus\Theta_{k}^{ N_2 }$, by the definition of $\Theta_k^{N_2}$ and $\Theta_n \subset \Lambda_{L_{n - 1}, L_n}$, there is $n_x\in\{k,k+1,\dots, N_2 \}$, s.t.
  $x\in\Lambda_{L_{n_x-1},L_{n_x}}\setminus\Theta_{n_x}$. Then by definition of singular set (Remark \ref{rmk-singularset}), there exists a  $(\omega,E_{n_x},m_0,\eta_0)$-good box $\Lambda_{(L_{n_x-1})_j}$ containing $x$ for some $j\in{1,2,\dots, N_2 }$, where $(L_{n_x-1})_j=L^{\rho^{n_x+j-1}}$.
  Since $$|E-E_{n_x}|\leq e^{-\frac{30M}{K}L_{n_x}}\leq e^{-\frac{30M}{K}L^{\rho^{n_x}}}\leq e^{-\frac{30M}{K}(L_{n_x})_j},$$
 $\Lambda_{(L_{n_x-1})_j}$ is also $(\omega,E,\frac{M}{K},\eta_0)$-good by Lemma \ref{lemma: stability}. Let $\phi_E$ be the normalized eigenfunction of $E$ on $H_{\omega,\Lambda_{L_{k-1}}}$, Then
  \[
  |\phi_E(x)|\leq e^{-\frac{M}{K}L^{\rho^{n_x+j-1}}}\leq e^{-\frac{M}{K}L^{\rho^{2 N_2 -1}}}
  \]
  So we have
  \[
    \sum_{x\in\Theta_k^{ N_2 }}|\phi_E(x)|^2=1-\sum_{x\in\Lambda_{L_{k-1}}\setminus\Theta_k^{ N_2 }}|\phi_E(x)|^2\geq 1- CL^{\beta d}e^{-m'L^{\rho^{2 N_2 -1}}}\geq 1/2
  \]
  when $L$ is large enough. Notice that 
 \[
 \#B_{k-1}(\bfek)\sum_{x\in\Theta_k^{ N_2 }}|\phi_E(x)|^2\leq \Tr(\1_{\Theta_k^{ N_2 }}P_{\mathcal{I}}(H_{\omega,\Lambda_{L_{k-1}}}))\leq |\Theta_k^{ N_2 }|\leq L^{\beta d},
 \]
 where the first inequality follows from computing the trace with respect to the eigenbasis of $H_{\omega, \Lambda_{L_{k - 1}}}$:
 \[
 \begin{split}
    \Tr(\1_{\Theta_k^{ N_2 }}P_{\mathcal{I}}(H_{\omega,\Lambda_{L_{k-1}}})) &= \sum_{E\in \sigma^{(\mathcal I)}(H_{\omega, \Lambda_{L_{k - 1}}})} \langle \1_{\Theta_k^{N_2}} \phi_E, \phi_E\rangle\\
    &= \sum_{E\in \sigma^{(\mathcal I)}(H_{\omega, \Lambda_{L_{k - 1}}})} \Vert\1_{\Theta_k^{N_2}}\phi_E\Vert^2\\
    &\geq \#B_{k-1}(\bfek)\sum_{x\in\Theta_k^{ N_2 }}|\phi_E(x)|^2.
 \end{split}
 \]
 The second inequality follows from computing the trace with respect to $\delta_x$:
 \[
 \begin{split}
    \Tr(\1_{\Theta_k^{ N_2 }}P_{\mathcal{I}}(H_{\omega,\Lambda_{L_{k-1}}})) &=  \sum_{x} \langle \1_{\Theta_k^{N_2}} P_{\mathcal{I}}(H_{\omega,\Lambda_{L_{k-1}}}) \delta_x, \delta_x\rangle\\
    &\leq \sum_{x\in \Theta_k^{N_2}} \Vert \1_{\Theta_k^{N_2}} P_{\mathcal{I}}(H_{\omega,\Lambda_{L_{k-1}}})\Vert \cdot 1\leq |\Theta_k^{N_2}|. 
    \end{split}
 \]
 Thus $\#B_{k-1}(\bfek)\leq 2L^{\beta d}$. Using the recursion relation \eqref{eq-7f} from layer $L_{ N_2 }$ to layer $L_1$, we have
 \[
 \#D_0^{ N_2 }\leq CL_{ N_2 }^d(L^{\beta d})^{ N_2 }\leq CL^{( N_2 +1)\beta d}.
 \]
\end{proof}

\begin{proof}[Proof of Theorem \ref{thm: Key Theorem}]
    Recall that $\beta$, $\rho$, $N_1$, $N_2$, $M$ are given in \eqref{eq-7g}, \eqref{eq-7d}. Now given any $p<p'<p_0$, we can pick $b=1+\frac{1}{\beta}(p'-( N_2 +1)\beta)$, fix a $K \geq K_{d, b, p}$. Finally, set $c = \frac{15M}{K}$, $\mu = \frac{\beta}{2}$. 
    
    By Theorem \ref{side_2}, we get event $\mathcal S_{L, x_0}$. Meanwhile, set
    \begin{equation}
        \label{eq-7h}
    \mathcal R_{L, x_0}:= \bigcap_{E\in \sigma^{(\mathcal I, \red)}(H_{\omega, \Lambda})} \mathcal R_{L, x_0}^{(E)}
    \end{equation}
    where $\mathcal R_{L, x_0}^{(E)}$ is obtained from Proposition \ref{lemma: fixed_E0} with $E_0 = E$, $p' = p'$, $m = \frac{15M}{K}< m' = \frac{30M}{K}<m_0$. The set 
    \[
    \mathcal K_{L, x_0} = \mathcal S_{L, x_0} \cap \mathcal R_{L, x_0}
    \]
    suffices our needs. Indeed, 
    \[
    \PP(\mathcal K_{L, x_0}) \geq 1 - L^{-b\beta d} - L^{( N_2  + 1)\beta d} L^{-p'd} \geq 1 - L^{-pd}
    \]
    when $L$ is large enough. Furthermore, assume  $\omega \in \mathcal K_{L, x_0}$. Note that condition \eqref{Wmain} implies
    \[
    \Wxoe> e^{-cL^\mu}\geq e^{-\frac{15M}{K}L^{\frac{\beta}{2}}} \geq  e^{-30M \sqrt{\frac{L^\beta}{K}}},
    \]
    Since $\omega\in \mathcal S_{L, x_0}$, by Theorem \ref{side_2}, 
    \[
    \Wxoe> e^{-cL^\mu} \quad \Rightarrow \quad d(E, \sigma^{(\mathcal I, \red)}(H_{\omega, \Lambda_L(x_0)}))<e^{-\frac{30M}{K}L}.
    \]
    Since  $\omega\in \mathcal R_{L, x_0}$, by Proposition \ref{lemma: fixed_E0}, 
    \[
    d(E, \sigma^{(\mathcal I, \red)}(H_{\omega, \Lambda_L(x_0)}))<e^{-\frac{30M}{K}L} \quad \Rightarrow \quad \Wlxoe\leq e^{-\frac{15M}{K}L} = e^{-cL}.
    \]
    This completes the proof.
\end{proof}
\begin{remark}\label{rmk-7b}
    If $p_0>1$, then one only need the first spectral reduction but not the second. Basically take $\mathcal K_{L, x_0} = \mathcal O_{L, x_0} \cap \mathcal R_{L, x_0}$ is good enough, where $\mathcal O_{L, x_0}$ is obtained in Theorem \ref{side_1} and $\mathcal R_{L, x_0}$ is given above. 
\end{remark}

\section{Acknowledgments}
The authors would like to thank Abel Klein for explaining the subtleties of MSA arguments and providing invaluable commentary on this work and Svetlana Jitomirskaya for providing her encouragement, support, and insight throughout this project. They would also like to thank University of California, Irvine where most of the work was done when the authors were graduate students there. N.R. and X.Z. are partially supported by Simons 681675, NSF DMS-2052899 and DMS-2155211. X.Z. is also partially supported by NSF DMS-2054589 and Simons Foundation Targeted Grant (917524) to the Pacific Institute for the Mathematical Sciences.
\bibliographystyle{amsxport}
\bibliography{mybib1}

\begin{bibdiv}
\begin{biblist}

\bib{BK05}{article}{
      author={Bourgain, Jean},
      author={Kenig, Carlos~E},
       title={On localization in the continuous {A}nderson-{B}ernoulli model in
  higher dimension},
        date={2005},
     journal={Inventiones mathematicae},
      volume={161},
       pages={389\ndash 426},
}

\bib{B96II}{article}{
      author={Berezansky, Yurij~M},
      author={Sheftel, Zinovij~G},
      author={Us, GF},
       title={Functional analysis. vol. ii, volume 86 of},
        date={1996},
     journal={Operator Theory: Advances and Applications},
}

\bib{CKM87}{article}{
      author={Carmona, Ren{\'e}},
      author={Klein, Abel},
      author={Martinelli, Fabio},
       title={{A}nderson localization for bernoulli and other singular
  potentials},
        date={1987},
     journal={Communications in Mathematical Physics},
      volume={108},
      number={1},
       pages={41\ndash 66},
}

\bib{CS87}{book}{
      author={Cycon, Hans~L},
      author={Simon, Barry},
       title={{S}chr{\"o}dinger operators: with applications to quantum
  mechanics and global geometry},
   publisher={Springer Science \& Business Media},
        date={1987},
}

\bib{DJLS95}{article}{
      author={del Rio, Rafael},
      author={Jitomirskaya, Svetlana},
      author={Last, Yoram},
      author={Simon, Barry},
       title={What is localization?},
        date={1995},
     journal={Physical review letters},
      volume={75},
      number={1},
       pages={117},
}

\bib{DJLS96}{article}{
      author={del Rio, Rafael},
      author={Jitomirskaya, Svetlana},
      author={Last, Yoram},
      author={Simon, Barry},
       title={Operators with singular continuous spectrum, iv. {H}ausdorff
  dimensions, rank one perturbations, and localization},
        date={1996},
     journal={Journal d’Analyse Math{\'e}matique},
      volume={69},
      number={1},
       pages={153\ndash 200},
}

\bib{DS01}{article}{
      author={Damanik, David},
      author={Stollmann, Peter},
       title={Multi-scale analysis implies strong dynamical localization},
        date={2001},
     journal={Geometric \& Functional Analysis GAFA},
      volume={11},
      number={1},
       pages={11\ndash 29},
}

\bib{DS18}{article}{
      author={Ding, Jian},
      author={Smart, Charles~K},
       title={Localization near the edge for the {A}nderson {B}ernoulli model
  on the two dimensional lattice},
        date={2018},
     journal={Inventiones mathematicae},
       pages={1\ndash 40},
}

\bib{FS84}{article}{
      author={Fr{\"o}hlich, J{\"u}rg},
      author={Spencer, Thomas},
       title={A rigorous approach to {A}nderson localization},
        date={1984},
     journal={Physics Reports},
      volume={103},
      number={1-4},
       pages={9\ndash 25},
}

\bib{GD98}{article}{
      author={Germinet, Fran{\c{c}}ois},
      author={De~Bievre, Stephan},
       title={Dynamical localization for discrete and continuous random
  {S}chr{\"o}dinger operators},
        date={1998},
     journal={Communications in mathematical physics},
      volume={194},
      number={2},
       pages={323\ndash 341},
}

\bib{G99}{article}{
      author={Germinet, Fran{\c{c}}ois},
       title={Dynamical localization {II} with an application to the almost
  {M}athieu operator},
        date={1999},
     journal={Journal of statistical physics},
      volume={95},
      number={1-2},
       pages={273\ndash 286},
}

\bib{GJ01}{article}{
      author={Germinet, Fran{\c{c}}ois},
      author={Jitomirskaya, Svetlana},
       title={Strong dynamical localization for the almost {M}athieu model},
        date={2001},
     journal={Reviews in Mathematical Physics},
      volume={13},
      number={06},
       pages={755\ndash 765},
}

\bib{GK01}{article}{
      author={Germinet, Fran{\c{c}}ois},
      author={Klein, Abel},
       title={Bootstrap multiscale analysis and localization in random media},
        date={2001},
     journal={Communications in Mathematical Physics},
      volume={222},
      number={2},
       pages={415\ndash 448},
}

\bib{GK06}{article}{
      author={Germinet, Francois},
      author={Klein, Abel},
       title={New characterizations of the region of complete localization for
  random {S}chr{\"o}dinger operators},
        date={2006},
     journal={Journal of statistical physics},
      volume={122},
       pages={73\ndash 94},
}

\bib{JZ19}{article}{
      author={Jitomirskaya, Svetlana},
      author={Zhu, Xiaowen},
       title={Large deviations of the lyapunov exponent and localization for
  the 1d {A}nderson model},
        date={2019},
     journal={Communications in Mathematical Physics},
      volume={370},
       pages={311\ndash 324},
}

\bib{GK12}{article}{
      author={Klein, Abel},
      author={Germinet, Fran{\c{c}}ois},
       title={A comprehensive proof of localization for continuous {A}nderson
  models with singular random potentials},
        date={2012},
     journal={Journal of the European Mathematical Society},
      volume={15},
      number={1},
       pages={53\ndash 143},
}

\bib{K07}{article}{
      author={Kirsch, Werner},
       title={An invitation to random {S}chr{\"o}dinger operators},
        date={2007},
     journal={arXiv preprint arXiv:0709.3707},
}

\bib{KlKS02}{article}{
      author={Klein, Abel},
      author={Koines, Andrew},
      author={Seifert, Maximilian},
       title={Generalized eigenfunctions for waves in inhomogeneous media},
        date={2002},
     journal={Journal of Functional Analysis},
      volume={190},
      number={1},
       pages={255\ndash 291},
}

\bib{KS87}{article}{
      author={Khanin, KM},
      author={Sinai, Ya~G},
       title={A new proof of m. herman's theorem},
        date={1987},
     journal={Communications in Mathematical Physics},
      volume={112},
      number={1},
       pages={89\ndash 101},
}

\bib{LZ19}{article}{
      author={Li, Linjun},
      author={Zhang, Lingfu},
       title={{A}nderson--{B}ernoulli localization on the three-dimensional
  lattice and discrete unique continuation principle},
        date={2022},
     journal={Duke Mathematical Journal},
      volume={171},
      number={2},
       pages={327\ndash 415},
}

\bib{DK89}{article}{
      author={von Dreifus, Henrique},
      author={Klein, Abel},
       title={A new proof of localization in the {A}nderson tight binding
  model},
        date={1989},
     journal={Communications in Mathematical Physics},
      volume={124},
       pages={285\ndash 299},
}

\end{biblist}
\end{bibdiv}
\end{document}